\documentclass[11pt]{article}

\usepackage[utf8]{inputenc}
\usepackage[T1]{fontenc}
\usepackage{lmodern}  
\usepackage{microtype}
\usepackage{amsmath}
\usepackage{amssymb}
\usepackage{amsthm}
\usepackage{thmtools, thm-restate}  
\usepackage{bm}
\usepackage{dsfont}
\usepackage{authblk}
\usepackage{fullpage,parskip}
\usepackage{comment}
\usepackage{wrapfig}
\usepackage{tikz}
\usepackage{mathtools}
\usepackage{bbm}
\usepackage[shortlabels]{enumitem}
\usepackage[
    backend=biber,
    style=alphabetic,
    maxcitenames=10,
    maxbibnames=99,
    natbib=true,
    url=false,
    doi=false, 
    isbn=false]{biblatex}
\usepackage{nag}
\usepackage{algpseudocode}
\usepackage{float} 

\usepackage[ruled,vlined]{algorithm2e}

\definecolor{custom_color}{RGB}{225,231,255}
\usepackage[framemethod=tikz]{mdframed}

\newcommand{\declarecolor}[2]{\definecolor{#1}{RGB}{#2}\expandafter\newcommand\csname #1\endcsname[1]{\textcolor{#1}{##1}}}
\declarecolor{White}{255, 255, 255}
\declarecolor{Black}{0, 0, 0}
\declarecolor{Maroon}{128, 0, 0}
\declarecolor{Coral}{255, 127, 80}
\declarecolor{LimeGreen}{50, 205, 50}
\declarecolor{DarkGreen}{0, 100, 0}
\declarecolor{Navy}{0, 0, 128}

\usepackage{hyperref}
\usepackage[capitalize, nameinlink]{cleveref}
\hypersetup{
	colorlinks=true,
	pdfpagemode=UseNone,
	citecolor=DarkGreen,
	linkcolor=Maroon,
	urlcolor=Navy,
	pdfstartview=FitW
}

\theoremstyle{plain}
\newtheorem{theorem}{Theorem}[section]
\newtheorem{lemma}[theorem]{Lemma}
\newtheorem{corollary}[theorem]{Corollary}

\newtheorem{fact}[theorem]{Fact}

\newtheorem{observation}[theorem]{Observation}
\newtheorem{claim}[theorem]{Claim}
\newtheorem{assumption}[theorem]{Assumption}

\theoremstyle{definition}
\newtheorem{definition}[theorem]{Definition}

\theoremstyle{remark}

\newenvironment{proofsketch}{\proof}{\endproof}  

\usetikzlibrary{shapes, patterns, decorations, fit, intersections}

\newcommand*{\Z}{{\mathbb{Z}}}

\newcommand*{\eps}{{\varepsilon}}
\newcommand*{\cI}{{\mathcal{I}}}

\newcommand*{\calP}{{\mathcal{P}}}

\newcommand{\alert}[1]{{\color{orange}#1}}

\newcommand{\tl}[1]{\widetilde{#1}}

\renewcommand{\vec}[1]{\bm{#1}}
\newcommand{\mat}[1]{\mathbf{#1}}

\let\poly\relax
\let\polylog\relax
\DeclareMathOperator{\poly}{poly}
\DeclareMathOperator{\polylog}{polylog}
\DeclareMathOperator{\tw}{tw}
\DeclareMathOperator{\dilation}{\mathtt{dilation}}
\DeclareMathOperator{\congestion}{\mathtt{congestion}}
\DeclareMathOperator{\congest}{\textsc{CONGEST}}
\DeclareMathOperator{\congclique}{\mathsf{CLIQUE}}
\DeclareMathOperator{\local}{\mathsf{LOCAL}}
\DeclareMathOperator{\pram}{\textsc{PRAM}}
\DeclareMathOperator{\ncc}{\textsc{NCC}}
\DeclareMathOperator{\specspars}{\textsc{SpectralSparsify}}
\DeclareMathOperator{\solver}{\textsc{Solver}}
\DeclareMathOperator{\build}{\textsc{BuildChain}}
\DeclareMathOperator{\randwalkschur}{\textsc{RandomWalkSchur}}
\DeclareMathOperator{\ultraspars}{\textsc{UltraSparsify}}
\DeclareMathOperator{\eliminate}{\textsc{Eliminate}}
\DeclareMathOperator{\approxSC}{\textsc{ApproxSC}}
\newcommand{\approxeps}{\approx_{\eps}}
\newcommand{\schur}{\mat{SC}}
\DeclareMathOperator{\hybrid}{\textsc{HYBRID}}
\newcommand{\supervertex}{S^{G \rightarrow \overline{G}}}
\newcommand{\tree}{T^{G \rightarrow \overline{G}}}
\newcommand{\overG}{\overline{G}}
\newcommand{\overD}{\overline{D}}
\newcommand{\overn}{\overline{n}}
\newcommand{\hatG}{\widehat{G}}

\DeclareMathOperator{\lev}{\vec{lev}}

\DeclareMathOperator{\res}{res}

\DeclareMathOperator{\SQ}{SQ}

\DeclarePairedDelimiterX{\card}[1]{\lvert}{\rvert}{#1}
\DeclarePairedDelimiterX{\abs}[1]{\lvert}{\rvert}{#1}
\DeclarePairedDelimiterX{\norm}[1]{\lVert}{\rVert}{#1}
\DeclarePairedDelimiterX{\tuple}[1]{\lparen}{\rparen}{#1}
\DeclarePairedDelimiterX{\parens}[1]{\lparen}{\rparen}{#1}
\DeclarePairedDelimiterX{\brackets}[1]{\lbrack}{\rbrack}{#1}
\DeclarePairedDelimiterX{\set}[1]\{\}{#1}
\let\Pr\relax
\DeclarePairedDelimiterXPP{\Pr}[1]{\mathbb{P}}[]{}{#1}
\DeclarePairedDelimiterXPP{\PrX}[2]{\mathbb{P}_{#1}}[]{}{#2}
\DeclarePairedDelimiterXPP{\Ex}[1]{\mathbb{E}}[]{}{#1}
\DeclarePairedDelimiterXPP{\ExX}[2]{\mathbb{E}_{#1}}[]{}{#2}

\bibliography{paper}
\vspace{-1cm}
\title{Almost Universally Optimal Distributed Laplacian Solvers via Low-Congestion Shortcuts\footnote{The author ordering was randomized using \url{https://www.aeaweb.org/journals/policies/random-author-order/generator}. 
It is requested that citations of this work list the authors separated by \texttt{\textbackslash textcircled\{r\}} instead of commas: Anagnostides \textcircled{r} Lenzen \textcircled{r} Haeupler \textcircled{r} Zuzic \textcircled{r} Gouleakis.}}

\author[1]{Ioannis Anagnostides} 
\author[2]{Christoph Lenzen}
\author[3]{Bernhard Haeupler}
\author[4]{Goran Zuzic}
\author[5]{Themis Gouleakis}
\affil[1]{Carnegie Mellon University}
\affil[2]{CISPA Helmholtz Center for Information Security}
\affil[3]{ETH Z\"urich \& Carnegie Mellon University}
\affil[4]{ETH Z\"urich}
\affil[5]{National University of Singapore}
\footnotetext{Funding acknowledgments are deferred to the next page.}

\date{}

\begin{document}

\maketitle

\pagenumbering{gobble}

\begin{abstract}
In this paper, we refine the (almost) \emph{existentially optimal} distributed Laplacian solver recently developed by Forster, Goranci, Liu, Peng, Sun, and Ye (FOCS `21) into an (almost) \emph{universally optimal} distributed Laplacian solver.

\medskip

Specifically, when the topology is known, we show that any Laplacian system on an $n$-node graph with \emph{shortcut quality} $\text{SQ}(G)$ can be solved within $n^{o(1)} \text{SQ}(G) \log(1/\varepsilon)$ rounds, where $\varepsilon$ is the required accuracy. This almost matches our lower bound which guarantees that any correct algorithm on $G$ requires $\widetilde{\Omega}(\text{SQ}(G))$ rounds, even for a crude solution with $\varepsilon \le 1/2$. Even in the unknown-topology case (i.e., standard CONGEST), the same bounds also hold in most networks of interest. Furthermore, conditional on conjectured improvements in state-of-the-art constructions of low-congestion shortcuts, the CONGEST results will match the known-topology ones.

\medskip

Moreover, following a recent line of work in distributed algorithms, we consider a hybrid communication model which enhances CONGEST with limited global power in the form of the node-capacitated clique (NCC) model. In this model, we show the existence of a Laplacian solver with round complexity $n^{o(1)} \log(1/\varepsilon)$.

\medskip

The unifying thread of these results, and our main technical contribution, is the study of novel \emph{congested} generalization of the standard \emph{part-wise aggregation} problem. We develop near-optimal algorithms for this primitive in the Supported-CONGEST model, almost-optimal algorithms in (standard) CONGEST, as well as a very simple algorithm for bounded-treewidth graphs with slightly worse bounds. This primitive can be readily used to accelerate the FOCS`21 Laplacian solver. We believe this primitive will find further independent applications.
\end{abstract}

\newpage

\footnotetext{Funding acknowledgments. Bernhard Haeupler: Supported in part by NSF grants CCF-1814603, CCF-1910588, NSF CAREER award CCF-1750808, a Sloan Research Fellowship, funding from the European Research Council (ERC) under the European Union's Horizon 2020 research and innovation program (ERC grant agreement 949272), and the Swiss National Foundation (project grant 200021-184735). Goran Zuzic: Supported in part by the Swiss National Foundation (project grant 200021-184735). Themis Gouleakis: Supported in part by an NRF Fellowship for AI (R-252-000-A33-133). Part of the work was done while visiting the Simons Institute for Theory of Computing.}

\tableofcontents
\thispagestyle{empty} 

\clearpage
\pagenumbering{arabic}

\section{Introduction}

The \emph{Laplacian paradigm} has emerged as one of the cornerstones of modern algorithmic graph theory. Integrating techniques from combinatorial optimization with powerful machinery from numerical linear algebra, it was originally pioneered in \cite{DBLP:journals/siammax/SpielmanT14} who established the first nearly-linear time solvers for a (linear) Laplacian system. Thereafter, there has been a considerable amount of interest in providing simpler and more efficient solvers~\cite{DBLP:journals/siamcomp/KoutisMP14,DBLP:conf/stoc/KelnerOSZ13,DBLP:conf/focs/KyngS16}. Indeed, this framework has led to some state of the art algorithms for a wide range of fundamental graph-theoretic problems; e.g., see  \cite{Axiotis21:Faster,DBLP:conf/focs/Madry16,DBLP:conf/soda/CohenMSV17,DBLP:conf/focs/BrandLNPSS0W20,DBLP:conf/soda/KelnerLOS14,DBLP:conf/soda/Peng16,DBLP:conf/focs/AxiotisMV20}, and references therein. In the distributed setting, a major breakthrough was very recently made in \cite{DBLP:journals/corr/abs-2012-15675}. In particular, the authors developed a distributed algorithm that solves any Laplacian system on an $n$-node graph after $n^{o(1)} (\sqrt{n} + D) \log(1/\eps)$ rounds of the standard $\congest$ model, where $D$ represents the hop-diameter of the underlying network and $\eps > 0$ is the error of the solver. Moreover, they showed that their algorithm is \emph{existentially optimal}, up to the $n^{o(1)}$ factor, establishing a lower bound of $\widetilde{\Omega}(\sqrt{n} + D)$ rounds via a reduction from the $s-t$ connectivity problem~\cite{10.1145/1993636.1993686}.

This \emph{existential} lower bound in the $\congest$ model of distributed computing should hardly come as any surprise. Indeed, it is well-known by now that a remarkably wide range of \emph{global} optimization problems, including minimum spanning tree (MST), minimum cut (Min-Cut), maximum flow, and single-source shortest paths (SSSP), require $\widetilde{\Omega}(\sqrt{n} + D)$ rounds\footnote{As usual, we use the notation $\widetilde{O}(\cdot)$ and $\widetilde{\Omega}(\cdot)$ to suppress polylogarithmic factors on $n$.} \cite{814597,10.1145/1007352.1007407,10.1145/1993636.1993686}. The same limitation generally applies to any non-trivial approximation and even under randomization. Nonetheless, these lower bounds are constructed on some pathological graph instances which arguably do not occur in practice. This begs the question: \emph{Can we obtain more refined performance guarantees based on the underlying topology of the communication network?} The framework of \emph{low-congestion shortcuts}, introduced by \cite{DBLP:conf/soda/GhaffariH16}, demonstrated that bypassing the notorious $\Omega(\sqrt{n})$ lower bound is possible: MST and Min-Cut on \emph{planar graphs} can be solved in $\widetilde{O}(D)$ rounds. This is crucial, given that in many graphs of practical significance the diameter is remarkably small; e.g., $D = \polylog(n)$ (as is folklore, this holds for most social networks), implying \emph{exponential improvements} over generic algorithms used for general graphs. In the context of the distributed Laplacian paradigm, we raise the following question:


\begin{quote}
\centering
\textit{Is there a faster distributed Laplacian solver under ``non-worst-case'' families of graphs in the $\congest$ model?}
\end{quote}

The only known technique in distributed computing for designing algorithms that go below the $\sqrt{n}$-bound is the low-congestion shortcut framework of Ghaffari and Haeupler~\cite{DBLP:conf/soda/GhaffariH16}, and its large ecosystem of tools built around it~\cite{DBLP:conf/podc/HaeuplerIZ16,DBLP:conf/wdag/HaeuplerIZ16,DBLP:conf/stoc/HaeuplerWZ21,DBLP:journals/corr/abs-2008-03091,zuzic2022universally,ghaffari2021universally,hopexpander2022}. However, the ``$\rho$-congested minor'' primitive introduced and extensively used in the novel distributed Laplacian solver~\cite{DBLP:journals/corr/abs-2012-15675} is out-of-reach from the current set of tools available in the low-congestion shortcut framework. We address this issue by introducing an analogous primitive called \emph{$\rho$-congested part-wise aggregation}, which greatly simplifies the interface used by \cite{DBLP:journals/corr/abs-2012-15675}. We then extend the low-congestion shortcut framework with new techniques that enables it to near-optimally solve this primitive: we provide both an algorithm that utilizes the very recent hop-constrained expander decompositions for shortcut construction~\cite{hopexpander2022} to solve the primitive in general graphs with a linear dependence on $\rho$, as well as a very simple algorithm with a quadratic $\rho$-dependence for bounded-treewidth graphs. Finally, we settle our original question in the positive by establishing that our new primitive can be readily used to accelerate the distributed Laplacian solver for non-worst-case topologies.

Specifically, we show our new techniques are sufficient to lift the existentially optimal algorithm~\cite{DBLP:journals/corr/abs-2012-15675} to a \emph{universally optimal} algorithm---modulo $n^{o(1)}$ factor inherent in the prior approach---for distributedly solving a Laplacian system, meaning that, \emph{for any topology}, our algorithm is essentially as fast as possible. In other words, for any graph, our algorithm almost matches the best possible (correct) algorithm for that graph. This result is unconditional in essentially all settings of interest (see \Cref{thm:laplacian-full-models} for details), but relies on conjectured improvements of current state-of-the-art constructions of low-congestion shortcuts to achieve unqualified universal optimality---like all other results in the area.

Furthermore, another concrete way of bypassing the $\widetilde{\Omega}(\sqrt{n} + D)$ lower bound, besides investigating non-worst-case families of graphs, is by enhancing the local communication network with a limited amount of \emph{global power}. Indeed, research concerning \emph{hybrid} networks was recently initiated in the realm of distributed algorithms~\cite{DBLP:conf/soda/AugustineHKSS20}, although networks combining different communication modes have already found numerous applications in real-life computing systems; as such, hybrid networks have been intensely studied in other areas of distributed computing (see \cite{CHEN201645,10.1145/1851182.1851222,KAR2018203}, and references therein). In this paper, we will enhance the standard $\congest$ model with the recently introduced \emph{node-capacitated clique} (henceforth $\ncc$)~\cite{DBLP:conf/spaa/AugustineGGHSKL19}. The latter model enables all-to-all communication, but with severe capacity restrictions for every node. The integration of these models will be referred to as the $\hybrid$ model for the rest of this work. This leads to the following central question:

\begin{quote}
\centering
\textit{Is there a faster distributed Laplacian solver in the $\hybrid$ model?}
\end{quote}

Our paper essentially settles this question by showing the same $\rho$-congested part-wise aggregation primitive can be efficiently solved in $\tilde{O}(\rho)$ rounds of $\ncc$, implying an almost optimal $n^{o(1)}$-round distributed algorithm for solving Laplacian systems in the $\hybrid$ model. A conceptual contribution of our approach is that we treat both $\congest$, Supported-$\congest$, and $\hybrid$ in a \emph{unified way} through the lens of the low-congestion shortcut framework, by designing our algorithm using high-level primitives and leaving the model-specific translations to the framework itself. We note that a similar unified view of PRAM (i.e., parallel) and $\congest$ (i.e., distributed) graph algorithms through the same lens has led to very recent breakthroughs on long-standing open problems for both of these settings~\cite{rozhon2022undirected}.

\subsection{Overview of our Contributions and Techniques}


The unifying thread and the main technical ingredient of our (almost) universally optimal distributed Laplacian solvers is a new fundamental communication primitive which we refer to as the \emph{congested part-wise aggregation problem}. Specifically, we develop near-optimal algorithms for solving this problem in the (Supported-)CONGEST and the NCC model (\Cref{section:congested-aggregation}), and then we utilize this primitive to develop almost universally optimal Laplacian solvers in \Cref{section:Laplacian}.

\subsubsection{The Congested Part-Wise Aggregation Problem}

To introduce the congested part-wise aggregation problem, let us first give some basic background. The aforementioned Ghaffari-Haeupler framework of low-congestion shortcuts revolves around the so-called \emph{part-wise aggregation problem} posed as follows: ``The graph is partitioned into \emph{disjoint} and individually-connected parts, and we need to compute some simple aggregate function for each part, e.g., the minimum of the values held by the nodes in a given part''~\citep{DBLP:conf/soda/GhaffariH16} (see \Cref{definition:part_wise_problem} for a formal definition). Importantly, it has been shown that this primitive can be solved efficiently in \emph{structured} topologies, and that many problems (including the MST, shortest path, min-cut, etc.) reduce to a small number of calls to a part-wise aggregation oracle, leading to universally optimal algorithms. Unfortunately, it is not clear how to reduce solving a Laplacian system to (a small number of) part-wise aggregation calls and in this paper, we primarily address this issue.

Our first technical contribution is to extend the framework of low-congestion shortcuts by studying a more general primitive: one that incorporates \emph{congestion} (of the input parts) into the underlying \emph{part-wise aggregation} instance. More precisely, unlike the standard part-wise aggregation problem, we allow each node to participate in up to $\rho \in \Z_{\geq 1}$ aggregation parts (see \Cref{definition:congested_part_wise_problem}). We later show that efficient solutions to this primitive leads to efficient distributed Laplacian solvers.

We first remark that a natural strategy for solving congested part-wise aggregation instances does not work: congested instances \emph{cannot}, in general, be directly reduced to a ``small'' collection of $1$-congested instances, thereby necessitating a more refined approach. To this end, our approach is based on ``lifting'' the underlying communication network $\overG$ into its $\rho$-\emph{layered version} $\hatG_{O(\rho)}$: every edge is replaced with a matching and every node with a $\rho$-clique. The importance of this transformation is that, as we show in \Cref{pa-layered-to-pa-congested}, the $\rho$-congested part-wise aggregation problem can be reduced to a $1$-congested instance on the $\rho$-layered graph (\Cref{sec:layered-graph}). This is first established under the assumption that individual parts correspond to simple paths, and then we extend our results to general parts by following \citet{DBLP:conf/stoc/HaeuplerWZ21}. In light of this reduction, we next focus on solving the $1$-congested part-wise aggregation instance on the layered graph.

As a warm-up, we treat graphs with bounded \emph{treewidth} $\tw(G)$ (\Cref{definition:treewidth}). It is known from \cite{DBLP:conf/wdag/HaeuplerIZ16} that on a graph $G$ with treewidth $\tw(G)$, a $1$-congested part-wise aggregation instance can be solved in $\tl{O}(\tw(G) D)$ rounds of CONGEST. Keeping this in mind, we first show that the treewidth of the $\rho$-layered graph $\hatG_{\rho}$ can only increase by a factor of $\rho$ compared to the original graph (\Cref{lemma:simulated_density}). Hence, we can solve $1$-congested instances in $\hatG_{O(\rho)}$ in $\tl{O}(\rho \tw(\overG) D)$ rounds (when the underlying network is $\hatG_{O(\rho)}$), which in turn allows us to solve $\rho$-congested instances on $\overG$ in $\tl{O}(\rho^2 \tw(G) D)$ time in $G$ (another $\rho$ factor is necessary to simulate $\hatG_{O(\rho)}$ in $\overG$). This positive result poses a natural question: can we achieve similar results on graphs with bounded \emph{minor density} $\delta(G)$ (\Cref{def:minor-density})? 
However, the answer to this question is negative: \emph{minor density} can blow up even for a $2$-layered planar graph (see \Cref{observation:minordensity}), making such a result impossible. 

Then, we look at arbitrary graphs $G$: it is known that $1$-congested part-wise aggregation instances can be solved in a number of rounds that is controlled by $\SQ(G)$, where $\SQ(G)$ is the \emph{shortcut quality} of $G$ (a certain graph parameter we formalize in \Cref{def:shortcut-quality}). Specifically, it can be solved in $\tl{O}(\SQ(G))$ rounds when the topology is known in advance\footnote{This model is also known as the \emph{supported} $\congest$. That is, $\congest$ under the assumption that the topology is known; see \Cref{sec:prel} for a formal description of the model. Our techniques also apply in the full generality of $\congest$, as we explain in the sequel.}~\cite{DBLP:conf/stoc/HaeuplerWZ21} and $\poly(\SQ(G)) \cdot n^{o(1)}$ in general CONGEST~\cite{hopexpander2022}. The shortcut quality parameter is significant because it was shown that many distributed problems (including the MST, shortest path, min-cut, and---Laplacian solving, as we show later) require $\tl{\Omega}(\SQ(G))$ rounds in CONGEST to be solved on $G$~\cite{DBLP:conf/stoc/HaeuplerWZ21}. Therefore, algorithms that have an upper bound close to $\SQ(G)$ are \emph{universally optimal}.



With the end goal of solving the $1$-congested part-wise aggregations on layered graphs $\hatG_\rho$ in time controlled by $\SQ(G)$, our main result established that \emph{the shortcut quality of the $\rho$-layered graph does not increase} (modulo polylogarithmic factors) as compared to the original graph (\Cref{theorem:Grho-small-shortcut-quality}). This has a plethora of important consequences: (1) when $\SQ(G) \le n^{o(1)}$, we can unconditionally solve $\rho$-congested part-wise aggregation instances in $\rho \cdot n^{o(1)}$ CONGEST rounds and (2) when the topology of $G$ is known, there exists a distributed algorithm which solves any $\rho$-congested part-wise aggregation problem in $\rho \cdot \widetilde{O}(\SQ(G))$ rounds. As a consequence of our general result, the shortcut quality of any $2$-layered planar graph is $\widetilde{O}(D)$ since it is known that the shortcut quality of a planar graph is $\tl{O}(D)$~\cite{DBLP:conf/soda/GhaffariH16}. This constitutes perhaps the most natural example of a graph whose minor density is very far from the shortcut quality; the only other example documented in the literature so far is that of \emph{expander graphs}.

Our proof proceeds by employing alternative characterizations of the shortcut quality in terms of certain communication tasks. Specifically, shortcut quality can be shown to be equal (modulo polylogarithmic factors) to the following two-player max-min game: the first (max) player chooses $k$ sources and $k$ sinks in the graph such that we can find $k$ node-disjoint paths matching the sources with the sinks; then the second (min) player finds the smallest so-called \emph{quality} $Q$ such that there exist $k$ paths matching the sources with the sinks with the path lengths being at most $Q$ and each edge of the underlying graph supporting at most $Q$ of second player's paths. This characterization allows us to compare the shortcut quality of $\hatG_{\rho}$ with $\overG$ as follows: take the worst-case (first player's) set of sources and sinks in $\hatG_\rho$. Project them to $\overG$ and note they have node congestion $\rho$ (due to the construction of $\hatG_{\rho}$). Then, we show we can decompose (i.e., partition) these set of sources and sinks into $\widetilde{O}(\rho)$ pairs of sub-sources and sub-sinks that are node-disjointly connectable in $G$. However, each such set enjoys paths of quality $\SQ(G)$, hence embedding each such pair in a separate layer of $\hatG_{\rho}$ shows that the shortcut quality of $\SQ(\hatG_\rho)$ is at most $\widetilde{O}(\SQ(\overG))$. Although this general approach improves over our result for treewidth-bounded graphs we previously described, our approach for the latter class of graphs is substantially simpler and more suited for potential practical applications.


\subsubsection{Almost Universally Optimal Laplacian Solvers}

First, we note that any distributed Laplacian solver that always correctly outputs an answer on a fixed graph $G$ must take at least $\tilde{\Omega}(\SQ(G))$ rounds, giving us a lower bound to compare ourselves with. Our refined lower bound uses the hardness result recently shown by \cite{DBLP:conf/stoc/HaeuplerWZ21} for the spanning connected subgraph problem, applicable for \emph{any} (i.e., non-worst-case) graph $G$. Specifically, we show that a Laplacian solver can be leveraged to solve the spanning connected subgraph problem, thereby substantially strengthening the lower bound in \cite{DBLP:journals/corr/abs-2012-15675}.
\begin{restatable}{proposition}{lowerb}
  \label{theorem:lower_bound}
  Consider a graph $\overline{G}$ with shortcut quality $\SQ(\overline{G})$. Then, solving a Laplacian system on $\overline{G}$ with $\eps \leq \frac{1}{2}$ requires $\widetilde{\Omega}(\SQ(\overline{G}))$ rounds in both $\congest$ and Supported-$\congest$ models.
\end{restatable}

On the upper-bound side, we utilize the congested part-wise aggregation primitive to improve and refine the Laplacian solver of \cite{DBLP:journals/corr/abs-2012-15675}, leading to a substantial improvement in the round complexity under \emph{structured} network topologies.

\begin{restatable}{theorem}{thmLaplacianFullModels}\label{thm:laplacian-full-models}
  Consider any $n$-node graph $G$ with shortcut quality $\SQ(G)$ and hop-diameter $D$. There exists a distributed Laplacian solver with error $\eps > 0$  with the following guarantees:
  \begin{itemize}
  \item In the Supported-$\congest$ model, it requires $n^{o(1)} \SQ(G) \log(1/\eps)$ rounds.
  \item In the $\congest$ model, it requires $n^{o(1)} \poly(\SQ(G)) \log(1/\eps)$ rounds.
  \item In the $\congest$ model on graphs with minor density $\delta$, it requires $n^{o(1)} \delta D \log(1/\eps)$ rounds.
  \end{itemize}
\end{restatable}


We note that the above algorithm is almost (up to inherent $n^{o(1)}$ factors) universally optimality for most settings of interest. Since it is (almost) matching the $\SQ(G)$-lower-bound, it is unconditionally universally optimal when the topology is known in advance (i.e., Supported-CONGEST). Furthermore, in standard CONGEST, we give almost universally optimal $D n^{o(1)} \log(1/\eps)$-round algorithms for topologies that include planar graphs, $n^{o(1)}$-\emph{genus} graphs, $n^{o(1)}$-treewidth graphs, excluded-minor graphs, since all of them are graphs with minor density $\delta(G) = n^{o(1)}$. Furthermore, for the realistic case of $D \le n^{o(1)}$, it holds for most networks of interest that $\SQ(G) \le n^{o(1)}$ (e.g., expanders, hop-constrained expanders, as well as all classes mentioned earlier), for which we get $n^{o(1)} \log(1/\eps)$-round solvers. Finally, the conjectured improvements of the state-of-the-art of almost-optimal low-congestion shortcut constructions would immediately lift our results to be unconditionally universally optimal in CONGEST. However, the issue is orthogonal and out-of-scope of this paper.

Furthermore, in $\hybrid$ we obtain an almost optimal complexity in \emph{general graphs}:

\begin{restatable}{theorem}{thmLaplacianHybrid}\label{thm:laplacian-hybrid}
Consider any $n$-node graph. There exists a distributed Laplacian solver in the $\hybrid$ model with round complexity $n^{o(1)} \log(1/\eps)$, where $\eps > 0$ is the error of the solver.
\end{restatable}

This implies a remarkably fast subroutine for solving a Laplacian system in $\hybrid$ under arbitrary topologies. As a result, we corroborate the observation that a very limited amount of global power can lead to substantially faster algorithms for certain optimization problems, supplementing a recent line of work~\cite{Censor-Hillel21:On,DBLP:conf/soda/AugustineHKSS20,10.1145/3382734.3405719,feldmann2020fast,censorhillel2020distance,gotte2020timeoptimal,Kuhn22:Routing,Coy22:Nearshortest}. Furthermore, our framework based on the congested part-wise aggregation problem allows for a unifying treatment of both (Supported-)$\congest$ and $\hybrid$, and we consider this to be an important conceptual contribution of our work. Indeed, as we previously explained, both of our accelerated Laplacian solvers rely on faster algorithms for solving the congested part-wise aggregation problem. In particular, for (Supported-)$\congest$ we have already described our approach in detail, while in the $\hybrid$ model we employ certain communication primitives developed in \cite{DBLP:conf/spaa/AugustineGGHSKL19} for dealing with congestion in part-wise aggregations. A byproduct of our results is that the framework of low-congestion shortcuts interacts particularly well with the $\hybrid$ model, as was also observed in~\cite{anagnostides2021deterministic}.


\subsection{Further Related Work}

Our main reference point is the recent Laplacian solver of \citet{DBLP:journals/corr/abs-2012-15675} with \emph{existentially} almost-optimal complexity of $n^{o(1)}(\sqrt{n} + D) \log(1/\eps)$ rounds, where $\eps > 0$ represents the error of the solver. Specifically, they devised several new ideas and techniques to circumvent certain issues which mostly relate to the bandwidth restrictions of the $\congest$ model; these building blocks, as well as the resulting Laplacian solver are revisited in our work to refine the performance of the solver. We are not aware of any previous research addressing this problem in the distributed context. On the other hand, the Laplacian paradigm has attracted a considerable amount of interest in the community of parallel algorithms. Most notably, we refer to \cite{DBLP:conf/stoc/PengS14,DBLP:journals/mst/BlellochGKMPT14}. These approaches in the $\pram$ model of parallel computing fail---at least without non-trivial modifications---to lead to a almost-optimal solver in the distributed context~\cite{DBLP:journals/corr/abs-2012-15675}.

In addition to being a problem of independent interest, solving Laplacian systems often leads to a plethora of very fast algorithms (albeit typically polynomially-away from being optimal) for other problems such as (exact) maximum flow~\cite{DBLP:conf/focs/Madry16}, min-cost flow~\cite{Axiotis21:Faster}, shortest paths with negative weights~\cite{DBLP:conf/soda/CohenMSV17}, etc. The recent distributed Laplacian solver~\cite{DBLP:journals/corr/abs-2012-15675} also contributed fast analogues of these algorithms in the distributed model. A natural question to ask is whether we can also use our techniques to make these algorithms work for more structured graphs. However, these algorithms rely on directed or exact shortest path computations, which currently represent a major barrier for shortcut-based approaches. Moreover, the same set of problems represent a barrier even for existentially-optimal approaches as the current state-of-the-art is a factor of $D^{1/4}$ away from achieving unqualified existential optimality~\cite{chechik2021single}.

Research concerning hybrid communication networks in distributed algorithms was recently initiated by \cite{DBLP:conf/soda/AugustineHKSS20}. Specifically, they investigated the power of a model which integrates the standard $\local$ model~\cite{10.1137/0221015} with the recently introduced node-capacitated clique ($\ncc$)~\cite{DBLP:conf/spaa/AugustineGGHSKL19}, focusing mostly on distance computation tasks. Several of their results were subsequently improved and strengthened in subsequent works~\cite{10.1145/3382734.3405719,censorhillel2020distance} under the same model of computation. In our work we consider a substantially weaker model, imposing a severe limitation on the communication over the ``local edges''. This particular variant has been already studied in some recent works for a variety of fundamental problems~\cite{feldmann2020fast,gotte2020timeoptimal}.

The $\ncc$ model, which captures the global network in all hybrid models studied thus far, was introduced in~\cite{DBLP:conf/spaa/AugustineGGHSKL19} partly to address the unrealistic power of the \emph{congested clique} ($\congclique$)~\cite{10.1145/777412.777428}. In the latter model each node can communicate \emph{concurrently and independently} with \emph{all} other nodes by $O(\log n)$-bit messages. In contrast, the $\ncc$ model allows communication with $O(\log n)$ (arbitrary) nodes per round. As a result, in the $\hybrid$ model and under a sparse local network, only $\widetilde{\Theta}(n)$ bits can be exchanged overall per round, whereas $\congclique$ allows for the exchange of up to $\widetilde{\Theta}(n^2)$ (distinct) bits. As evidence for the power of $\congclique$ we note that even slightly super-constant lower bounds would give new lower bounds in circuit complexity, as implied by a simulation argument in~\cite{DBLP:conf/podc/DruckerKO13}. 

\section{Preliminaries}
\label{sec:prel}

\paragraph{General notation} We denote with $[k] := \{ 1, 2, \ldots, k \}$. Graphs throughout this paper are undirected. The nodes and the edges of a given graph $G$ are denoted as $V(G)$ and $E(G)$, respectively. We also use $n := |V(G)|$ for brevity. The graphs are often weighted, in which case we assume (as is standard) that for all $e \in E(G), \vec{w}(e) \in \{1, 2, \dots, \poly(n)\}$. We will denote the hop-diameter of a graph $G$ with $D(G)$ (the hop-diameter ignores weights). Moreover, we use $A \uplus B$ to denote the multiset union, i.e., each element is repeated according to its multiplicity; this operation corresponds to disjoint unions when $A \cap B = \emptyset$.



\paragraph{Communication models} The communication network consists of a set of $\overline{n}$ entities with $[\overline{n}] := \{1, 2, \dots, \overline{n}\}$ being the set of their IDs, and a local communication \emph{topology} given by a graph $\overline{G}$.\footnote{To avoid any possible confusion we point out that, for consistency with the nomenclature of \citet{DBLP:journals/corr/abs-2012-15675}, we 
  henceforth reserve $\overline{G}$ to denote the underlying \emph{communication network}, while $G$ is used in statements regarding arbitrary graphs.} We define $D := D(\overG)$ to be the (hop-)diameter of the underlying network. At the beginning, each node knows its own unique $O(\log \overn)$-bit identifier as well as the weights of the incident edges. Communication occurs in \emph{synchronous rounds}, and in every round nodes have unlimited computational power to process the information they possess. We will consider models with both \emph{local} and \emph{global} communication modes.

  The \emph{local} communication mode will be modeled with the \emph{CONGEST model}~\cite{peleg2000distributed} and \emph{Supported-CONGEST model}~\cite{Schmid2013}, for which in each round every node can exchange an $O(\log \overline{n})$-bit message with each of its neighbors in $\overline{G}$ via the \emph{local} edges. In the (standard) $\congest$ model, each node $v \in V(\overG)$ initially only knows the identifiers of each node in $v$'s own neighborhood, but has no further knowledge about the topology of the graph. On the other hand, in the Supported-CONGEST model, all nodes know the entire topology of $\overG$ upfront, but not the input.

  The \emph{global} communication mode will be modeled using \emph{NCC}~\cite{DBLP:conf/spaa/AugustineGGHSKL19}, for which in each round every node can exchange $O(\log \overline{n})$-bit messages with $O(\log \overline{n})$ arbitrary nodes via \emph{global} edges. If the capacity of some channel is exceeded, i.e., too many messages are sent to the same node, it will only receive an \emph{arbitrary} (potentially adversarially selected) subset of the information based on the capacity of the network; the rest of the messages are dropped. In this context, we will let $\hybrid$ be the integration of $\congest$ and $\ncc$ (i.e., nodes have both a \emph{local} and a \emph{global} communication mode at their disposal).

  The performance of a distributed algorithm will be measured in terms of its \emph{round complexity}---the number of rounds required so that every node knows its part of the output. For randomized algorithms it will suffice to reach the desired state with high probability.\footnote{We say that an event holds with high probability if it occurs with probability at least $1 - 1/n^c$ for a (freely choosable) constant $c > 0$.} We will assume throughout this work that nodes have access to a common source of randomness; this comes without any essential loss of generality in our setting~\citep{10.1145/2767386.2767417}. When talking about a distributed algorithm for a specific problem (e.g., Laplacian solving, part-wise aggregation, etc.) we assume the input is appropriately \emph{distributedly stored} (i.e., each node will know its own part) and, upon termination, it will be required that the output is appropriately distributedly stored. The appropriate way to distributedly store the input and output will be explained in the problem definition.
 
\paragraph{Low-Congestion Shortcuts} A recurring scenario in distributed algorithms for global problems (e.g.\ MST) boils down to solving the following part-wise aggregation problem:

\begin{definition}[Part-Wise Aggregation Problem]
    \label{definition:part_wise_problem}
    Consider an $n$-node graph $G$ whose node set $V(G)$ is \emph{partitioned} into $k$ (disjoint) parts $P_1 \uplus \dots \uplus P_k \subseteq V(G)$ such that each induced subgraph $G[P_i]$ is \emph{connected}. In the \emph{part-wise aggregation} problem, each node $v \in V$ is given its part-ID (if any) and an $O(\log n)$-bit value $\vec{x}(v)$ as input. The goal is that, for every part $P_i$, all nodes in $P_i$ learn the part-wise aggregate $\bigoplus_{w \in P_i} \vec{x}(w)$, where $\bigoplus$ is an arbitrary pre-defined \emph{aggregation function}.
\end{definition}
%
%
Throughout this paper, we will assume that the aggregation function $\bigoplus$ is commutative and associative (e.g. min, sum, logical-AND), although this is not strictly needed (e.g., see \cite{ghaffari2021universally}). To give a concrete example, in the context of Boruvka's algorithm for the MST problem, determining the minimum-weight outgoing edge for each part is an instance of a part-wise aggregation problem with $\bigoplus := \min$. To solve such problems, \citet{DBLP:conf/soda/GhaffariH16} introduced a natural combinatorial graph structure which they refer to as \emph{low-congestion shortcuts}.

\begin{definition}[Low-Congestion Shortcuts]\label{def:shortcuts}
  Consider a graph $G$ whose node set $V(G)$ is \emph{partitioned} into $k$ (disjoint) parts $P_1 \uplus \dots \uplus P_k \subseteq V(G)$ such that each induced subgraph $G[P_i]$ is \emph{connected}. A collection of subgraphs $H_1, \dots, H_k$ is a \emph{shortcut} of $G$ with \emph{congestion} $c$ and \emph{dilation} $d$ if the following properties hold: (i) the (hop) diameter of each subgraph $G[P_i] \cup H_i$ is at most $d$, and (ii) every edge is included in at most $c$ many of the subgraphs $H_i$. The quantity $Q = c + d$ will be referred to as the \emph{quality} of the shortcut.
\end{definition}

Importantly, a shortcut of quality $Q$ allows us to solve the part-wise aggregation problem in $\tl{O}(Q)$ rounds of $\congest$, as formalized below. For self-sufficiency, we include the proof in \Cref{appendix:prel}. 

\begin{restatable}{proposition}{shortcutparts}
\label{shortcuts-solving-pa}
  Suppose that $P_1, \ldots, P_k$ is any part-wise aggregation instance in a communication network $\overG$. Given a shortcut of quality $Q$, we can solve with high probability the part-wise aggregation problem in $\tl{O}(Q)$ $\congest$ rounds.
\end{restatable}


\paragraph{Shortcut Quality and Construction of Shortcuts} Shortcut quality, introduced below, is a fundamental graph parameter that has been proven to characterize the complexity of many important problems in distributed computing.

\begin{definition}\label{def:shortcut-quality}
  Given a graph $G = (V, E)$, we define the \emph{shortcut quality} $\SQ(G)$ of $G$ as the optimal (smallest) shortcut quality of the worst-case partition of $V$ into disjoint and connected parts $P_1 \uplus P_2 \uplus \ldots \uplus P_k \subseteq V$.
\end{definition}

For fundamental problems such as MST, SSSP, and Min-Cut any correct algorithm requires $\tl{\Omega}(\SQ(\overG))$ rounds on any network $\overG$, even if we allow randomized solutions and (non-trivial) approximation factors. In fact, this limitation holds even when the network topology $\overG$ is known to all nodes in advance~\citep{DBLP:conf/stoc/HaeuplerWZ21}. We remark that $\widetilde{\Omega}(D(\overG)) \le \SQ(\overG) \le O(D(\overG) + \sqrt{\overn})$, and the upper bound is known to be tight in certain (pathological) worst-case graph instances. This explains the notorious (existential) $\widetilde{\Omega}(D + \sqrt{n})$ lower bound pervasive in distributed computing~\citep{10.1145/1993636.1993686}.

Moreover, \emph{assuming} fast distributed algorithms for constructing shortcuts of quality competitive with $\SQ(\overG)$, all of the aforementioned problems can be solved in $\tl{O}( \SQ(\overG))$ rounds~\citep{DBLP:conf/soda/GhaffariH16,zuzic2022universally,ghaffari2021universally}. However, the key issue here is the algorithmic construction of the shortcuts upon which the above papers rely. While there has been a lot of recent progress in this regard, current algorithms are quite complicated and have sub-optimal guarantees. We recall below these state-of-the-art $\SQ(\overG)$-competitive construction results.

\begin{theorem}\label{thm:big-boy-construction}
  There exists a distributed algorithm that, given any part-wise aggregation instance on any $\overn$-node graph $\overG$, computes with high probability a shortcut with the following guarantees:
  \begin{itemize}\setlength\itemsep{0em}
  \item In $\congest$, the shortcut has quality $\poly\!\left(\SQ(\overG)\right) \cdot \overline{n}^{o(1)}$ and the algorithm terminates in $\poly\!\left(\SQ(\overG)\right) \cdot \overline{n}^{o(1)}$ rounds~\citep{hopexpander2022}.
  \item In Supported-CONGEST, the shortcut has quality $\tl{O}(\SQ(\overG))$ and the algorithm terminates in $\tl{O}(\SQ(\overG))$ rounds~\citep{DBLP:conf/stoc/HaeuplerWZ21}.
  \end{itemize}
\end{theorem}

\paragraph{Universal Optimality} A distributed algorithm is said to be \emph{$\alpha$-universally optimal} if, on every network graph $\overG$, it is $\alpha$-competitive with the fastest correct algorithm on $\overG$~\citep{DBLP:conf/stoc/HaeuplerWZ21}. Even the existence of such algorithms is not at all clear as it would seem possible that vastly different algorithms are required to leverage the structure of different networks. Nevertheless, a remarkable consequence of \Cref{thm:big-boy-construction} is that in Supported-CONGEST we can design $\widetilde{O}(1)$-universally optimal algorithms for many fundamental optimization problems. Moreover, efficient shortcut construction is the only obstacle towards achieving these results in the full generality of $\congest$, which is an issue orthogonal and out of scope for this paper. Still, the aforementioned results are sufficient to design $\overline{n}^{o(1)}$-universally optimal algorithms on graphs that have shortcut quality $\SQ(\overG) = \overline{n}^{o(1)}$, as it is arguably the case in most networks of practical interest.



\paragraph{Graphs Excluding Dense Minors} It turns out that the crucial issue of efficient shortcut construction can be resolved with a near-optimal, simple, and even deterministic algorithm for the rich class of graphs with \emph{bounded minor density}. Formally, let us first recall the following definition.\footnote{See the first part of \Cref{definition:low_congestion-minor} for a formal description of a minor.}
%
%
%

\begin{definition}[Minor Density]\label{def:minor-density}
The \emph{minor density} $\delta(G)$ of a graph $G$ is defined as 
\begin{equation*}
    \delta(G) = \max \left\{ \frac{|E'|}{|V'|} : H = (V', E') \textrm{ is a minor of $G$} \right\}.
\end{equation*}
\end{definition}
It should be noted that $\delta(G) = \widetilde{\Theta}(r(G))$, where $r(G)$ is the \emph{complete-graph minor size}, i.e., $r(G) = \max \{r : K_r \textrm{ is a minor of $G$} \}$~\citep{thomason_1984,THOMASON2001318}. Furthermore, any family of graphs closed under taking minors (such as planar graphs) has a constant minor density. For such graphs, \citet{DBLP:journals/corr/abs-2008-03091} established efficient shortcut construction: 
\begin{theorem}[\cite{DBLP:journals/corr/abs-2008-03091}]
    \label{theorem:quality-existence}
    Any graph $G$ with hop-diameter $D$ and minor density $\delta(G)$ admits shortcuts of quality $\tl{O}(\delta D)$, which can be constructed with high probability in $\tl{O}(\delta D)$ rounds of $\congest$.
\end{theorem}

The (linear) dependency on the minor density is existentially optimal~\citep[Lemma 3.2]{DBLP:journals/corr/abs-2008-03091}. It should be noted that, in the context of \Cref{theorem:quality-existence}, there is also a deterministic distributed algorithm with a slightly worse guarantee~\citep{DBLP:journals/corr/abs-2008-03091}. Some of our results apply for communication networks with \emph{bounded treewidth}, so let us recall the following definition.

\begin{definition}[Tree Decomposition and Treewidth]
    \label{definition:treewidth}
    A \emph{tree decomposition} of a graph $G$ is a tree $T$ with tree-nodes $X_1, \dots, X_k$, where each $X_i$ is a subset of $V(G)$ satisfying the following properties:
    \begin{enumerate}
        \item $ V = \bigcup_{i=1}^k X_i$;
        \item For any node $u \in V(G)$, the tree-nodes containing $u$ form a connected subtree of $T$;
        \item For every edge $\{u, v\} \in E(G)$, there exists a tree-node $X_i$ which contains both $u$ and $v$.
    \end{enumerate}
    The \emph{width} $w$ of the tree decomposition is defined as $w := \max_{i \in [k]} |X_i| - 1$. Moreover, the \emph{treewidth} $\tw(G)$ of $G$ is defined as the minimum of the width among all possible tree decompositions of $G$.
\end{definition}

Bounded-treewidth graphs inherit all of the nice properties guaranteed by \Cref{theorem:quality-existence}, as implied by the following well-known fact.

\begin{fact}
    \label{fact:minor_density} 
    For any graph $G$, $\delta(G) \leq \tw(G)$.
\end{fact}



\section{The Congested Part-Wise Aggregation Problem}
\label{section:congested-aggregation}


This section is concerned with a \emph{congested} generalization of the standard part-wise aggregation problem (\Cref{definition:part_wise_problem}), formally introduced below.

\begin{definition}[Congested Part-Wise Aggregation Problem]\label{definition:congested_part_wise_problem}
  Consider an $n$-node graph $G$ with a collection of $k$ subsets of nodes $P_1, \ldots, P_k \subseteq V(G)$ called \emph{parts} such that each induced subgraph $G[P_i]$ is \emph{connected} and each node $v \in V(G)$ is contained in at most $\rho \in \mathbb{Z}_{\ge 1}$ many parts, i.e., $\forall v \in V(G)\ \ |\{ i : P_i \ni v \}| \le \rho$. In the \emph{$\rho$-congested part-wise aggregation} problem, each node $v$ is given the following as input: for each part $P_i \ni v$ node $v$ knows the part-ID $i$ and an $O(\log n)$-bit part-specific value $\vec{x}_i(v)$. The goal is that, for each part $P_i$, all nodes in $P_i$ learn the part-wise aggregate $\bigoplus_{w \in P_i} \vec{x}_i(w)$, where $\bigoplus$ is an arbitrary pre-defined \emph{aggregation function}.
\end{definition}

This congested generalization of the standard part-wise aggregation problem that we study in this section turns out to be a central ingredient in our refined Laplacian solver; this is further explained in \Cref{section:Laplacian}. 
The remainder of this section is organized as follows. In \Cref{subsection:congest} we establish near-optimal algorithms for solving congested part-wise aggregations in $\congest$, which is also the main focus of this section. We conclude by pointing out the construction for $\ncc$ in \Cref{subsection:ncc}.

\subsection{Solving Congested Instances in the CONGEST Model}
\label{subsection:congest}

The first natural strategy for solving the $\rho$-congested part-wise aggregation problem of \Cref{definition:congested_part_wise_problem} is through a reduction to $\poly(\rho)$ $1$-congested instances. However, this approach immediately fails even if we allow $\rho = 2$. Indeed, there exist congested part-wise aggregation instances for which every two (distinct) parts share a common node, even when $\rho = 2$, leading to the following observation.

\begin{observation}
    For an infinite family of values $\overline{n}$, there exists an $\overline{n}$-node planar graph $\overline{G}$ and a $2$-congested part-wise aggregation instance $\cI$ with $k = \Theta(\sqrt{\overline{n}})$ parts such that reducing $\cI$ to the union of $k'$ $1$-congested part-wise aggregation instances on $\overline{G}$ requires $k' = \Omega(\sqrt{\overline{n}})$.
\end{observation}
\begin{wrapfigure}{R}{8cm}
  \centering
    \includegraphics[scale=0.45]{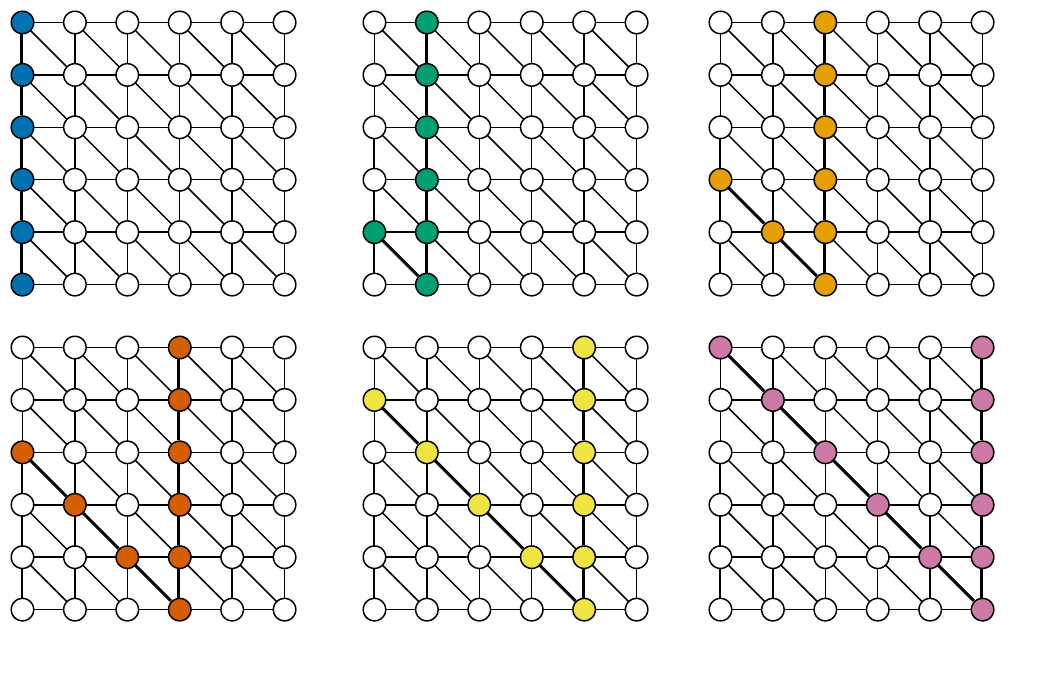}
    \caption{A $2$-congested part-wise aggregation problem on a $6 \times 6$ grid (the instance immediately extends to a $\sqrt{\overline{n}} \times \sqrt{\overline{n}}$ topology). Different colors highlight different parts of the instance.} 
    \label{fig:union}
\end{wrapfigure}
Such a pattern is illustrated in \Cref{fig:union}. As a result, directly employing a $1$-congested part-wise aggregation oracle is of little use since it would introduce an overhead depending on the number of parts. In light of this, we develop a more refined approach that leverages what we refer to as the \emph{layered graph}. This concept is introduced in \Cref{sec:layered-graph}, where we show that the congested part-wise aggregation problem can be reduced to the $1$-congested part-wise aggregation problem in the layered graph. Then, we give an algorithm for the $\rho$-congested part-wise aggregation problem in treewidth-bounded graphs through a simple approach in \Cref{sec:treewith-partwise}, yielding an $\tl{O}(\rho^2 \tw(\overG) D)$-round algorithm. Finally, we show that the shortcut quality $\SQ$ of the $\rho$-layered graph does not increase (modulo polylogarithmic factors) as compared to the original graph (\Cref{theorem:Grho-small-shortcut-quality}). This implies a solution for $\rho$-congested part-wise aggregations in general graphs with a runtime with the optimal, linear, dependence on $\rho$, albeit at the cost of a more involved argument (\Cref{sec:general-partwise}, specifically \Cref{congested-pa-general-solving}).

\subsubsection{The Layered Graph}\label{sec:layered-graph}

Here we introduce the \emph{layered graph} $\hatG_{\rho}$ associated with the underlying graph $\overG$. Then, we reduce the problem of $\rho$-congested part-wise aggregation on $\overG$ to a $1$-congested instance on $\hatG_{O(\rho)}$.

\begin{wrapfigure}{R}{8cm}
    \centering
    \includegraphics[scale=0.4]{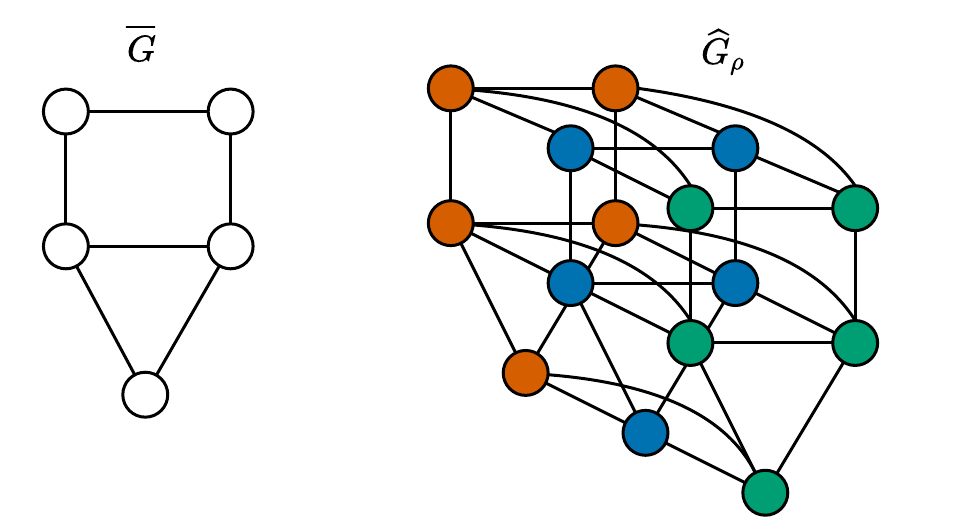}
    \caption{An example of a transformation from $\overline{G}$ to the layered graph $\widehat{G}_{\rho}$ with $\rho = 3$. We have highlighted with different colors different layers of the graph.}
    \label{fig:transformation}
\end{wrapfigure}

\paragraph{The Layered Graph} Consider an underlying network $\overline{G}$ and some $\rho \in \mathbb{Z}_{\ge 1}$, corresponding to the congestion parameter in \Cref{definition:congested_part_wise_problem}. The \emph{layered graph} $\hatG_\rho$ is constructed in the following way. First, we let $\hatG_\rho$ be a disjoint union of $\rho$ copies of $\overG$ (called \emph{layers}), namely $\overG_1, \overG_2, \ldots, \overG_{\rho}$. Each node $v \in V(\overG)$ is associated with its copies $v_1, v_2, \ldots, v_\rho \in V(\hatG_\rho)$. We also add an edge between each two copies that originate from the same node (i.e., we add a clique to $\hatG_\rho$ on the set of copies associated with the same node $v \in V(\overG)$); this construction is illustrated in \Cref{fig:transformation}. The layered graph induces a natural \emph{projection} operation $\pi : V(\hatG_{\rho}) \to V(\overG)$ which maps a copy $v_i$ to its original node $v = \pi(v_i)$. Furthermore, we often talk about simulating $\hatG_{\rho}$ in $\overG$, by which we mean that each node $v$ simulates---learns all the inputs and can generate all outputs---for its copies $v_1, \ldots, v_\rho$. Throughout this paper, we will assume that $\rho = \poly(\overline{n})$ so that any $O(\log n)$-bit message on $\hatG_{\rho}$ can be sent within $O(1)$ rounds in $\overG$; this also keeps the $\tl{O}$-notation well-defined.

The main goal of this section is to establish that the $\rho$-congested part-wise aggregation problem on $\overG$ can be reduced to a $1$-congested instance on $\widehat{G}_{O(\rho)}$, as formalized below.

\begin{restatable}[Unrestricted Congested Part-Wise Aggregation]{lemma}{paLayeredToPaCongested}\label{pa-layered-to-pa-congested}
  Let $\overG$ be an $\overn$-node graph and let $\mathbb{Z}_{\ge 1} \ni \rho \le \poly(\overn)$. Suppose that any ($1$-congested) part-wise aggregation on $\hatG_{O(\rho)}$ can be solved with a $\tau$-round CONGEST algorithm on $\hatG_{O(\rho)}$. Then, there exists an $\tl{O}(\rho \cdot \tau)$-round CONGEST algorithm on $\overG$ that solves any $\rho$-congested part-wise aggregation instance on $\overG$.
\end{restatable}

The remainder of this section is dedicated to the proof of this result. We first point out that any $\congest$ algorithm on $\widehat{G}_{\rho}$ can be simulated with only a $\rho$ multiplicative overhead in the round complexity (see \Cref{appendix:proof-3}).

\begin{restatable}[Simulating $\hatG_{\rho}$ in $\overG$]{lemma}{simGrho}
  \label{lemma:simulating-Grho}
  For any $\overG$ and any $\mathbb{Z}_{\ge 1} \ni \rho \le \poly(\overn)$, we can simulate any $\tau$-round CONGEST algorithm on $\hatG_\rho$ with a $(\rho \cdot \tau)$-round CONGEST algorithm on $\overG$.
\end{restatable}

Furthermore, we will use a folklore result showing how to color a (multi)graph of maximum degree $\Delta$ in $O(\Delta)$ colors in $O(\log n)$ rounds of CONGEST. By multigraph here we simply mean that there can be multiple parallel edges between the same pair of nodes, and every such edge can carry an independent message per round. To keep the paper self-contained we provide a short sketch of the proof in \Cref{appendix:proof-3}.

\begin{restatable}[Folklore, \cite{johansson1999simple}]{fact}{Johan}
\label{fact:fast-distributed-edge-coloring}
  Given a (multi)graph $G$ with $n$ nodes and maximum degree $\Delta \le \poly(n)$, there exists a randomized CONGEST algorithm that colors the edges of $G$ with $O(\Delta)$ colors and completes in $O(\log n)$ rounds, with high probability. The coloring is proper, i.e., two edges that share an endpoint are assigned a different color.
\end{restatable}

Now we are ready to prove a version of our main reduction (\Cref{pa-layered-to-pa-congested}), but with the slightly twist that we restrict each part of the $\rho$-congested part-wise aggregation problem to be a simple path. 
This restriction will be removed later.

\begin{lemma}[Path-Restricted Congested Part-Wise Aggregation]
    \label{lemma:path-restr}
  Let $\overG$ be a $\overn$-node graph and let $\mathbb{Z}_{\ge 1} \ni \rho \le \poly(\overn)$. Suppose that there exists a $\tau$-round CONGEST algorithm solving the ($1$-congested) part-wise aggregation on $\hatG_{O(\rho)}$. Then, there exists an $\tl{O}(\rho \cdot \tau)$-round CONGEST algorithm on $\overG$ that solves any $\rho$-congested part-wise aggregation instance on $\overG$ when each part is \emph{restricted to be a simple path}\footnote{I.e., there exists a simple path traversing all the nodes of the part, and each node knows the corresponding incident edges of that path.} (nodes are not repeated in simple paths).
\end{lemma}
\begin{proof}
  Let $\calP = \{ P_1, P_2, \ldots, P_k \}$ be subsets of nodes in $\overG$ comprising the parts of some $\rho$-congested part-wise aggregation on $\overG$. We will construct paths $\calP' = \{ P'_1, P'_2, \ldots, P'_k \}$ in $\hatG_{O(\rho)}$ in a way that solving a part-wise aggregation on $\calP'$ corresponds to solving a $\rho$-congested part-wise aggregation on $\calP$.

Let $E_i$ be the set of edges of $\overG$ comprising the simple path traversing all the nodes in $P_i$, and consider the graph $G' := (V(\overG), \biguplus_{i=1}^k E_i)$. 
First, we observe that the degree of any node in $v \in V(G') = V(\overG)$ is at most $2 \rho$ since at most $\rho$ many parts contain $v$ and each part contributes at most $2$ to the degree (since $P_i$ is a simple path). Furthermore, we can simulate any $\psi$-round CONGEST algorithm on $G'$ with a $(\psi \cdot \rho)$-round CONGEST algorithm on $\overG$ as each edge $e \in E(\overG)$ appears at most $\rho$ times in $E(G')$ due to the part-wise aggregation instance being at most $\rho$-congested. Therefore, using \Cref{fact:fast-distributed-edge-coloring} we can distributedly color the edges of $G'$ into at most $O(\rho)$ colors in $O(\log n)$ CONGEST rounds on $G'$, which translates to $\tl{O}(\rho)$ CONGEST rounds on $\overG$. Suppose that the algorithm assigns a color $\vec{c}(e) \in \{ 1, \ldots, O(\rho) \}$ to each edge $e \in \biguplus_i E_i$.

We now construct $P'_i \subseteq \hatG_{O(\rho)}$ as follows: consider each edge $\{u, v\} \in E_i$ and add both $u_{\vec{c}(\{u, v\})}, v_{\vec{c}(\{u, v\})} \in V(\hatG_{O(\rho)})$ to $P'_i$ (i.e., the $\vec{c}(\{u,v\})$-th copy of both $u$ and $v$). By construction, $P'_i$ induces a connected subgraph and the projection $P'_i$ to $\overG$ is exactly $P_i$. Next, we invoke the ($1$-congested) part-wise aggregation $\tau$-round algorithm for $\{ \calP'_1, \ldots, \calP'_k \}$ on $\hatG_{\rho}$, which can be converted to an $\tl{O}(\tau \cdot \rho)$-round algorithm on $\overG$ (\Cref{lemma:simulating-Grho}). Thus, we obtain an $\tl{O}(\tau \cdot \rho)$-round CONGEST algorithm on $\overG$ which solves any path-restricted $\rho$-congested part-wise aggregation problem.
\end{proof}
Finally, our reduction in \Cref{pa-layered-to-pa-congested} follows by reformulating \citep[Lemma 7.2]{DBLP:conf/stoc/HaeuplerWZ21}, as we argue in \Cref{appendix:proof-3}.
%
%

\subsubsection{Treewidth-Bounded Graphs}
\label{sec:treewith-partwise}

Here we leverage the reduction we established in \Cref{pa-layered-to-pa-congested} to obtain a simple algorithm for solving the congested part-wise aggregation problem in treewidth-bounded graphs. The crucial observation is that the treewidth of the layered graph can only grow by a factor of $\rho$ compared to the treewidth of the underlying graph, as we show in \Cref{lemma:simulated_density}.

\begin{restatable}{claim}{triv}
  \label{claim:simulated_diameter}
  $D(\widehat{G}_{\rho}) \leq D(\overline{G}) + 1$.
\end{restatable}

\begin{lemma}
    \label{lemma:simulated_density}
    If the treewidth of $\overline{G}$ is $\tw(\overline{G})$, then $\tw(\widehat{G}_{\rho}) \leq \rho \tw(\overline{G}) + \rho - 1$.
\end{lemma}

\begin{proof}
Consider a tree decomposition (in the sense of \Cref{definition:treewidth}) of $\overline{G}$ into tree-nodes $\{X_j\}_{j=1}^k$ such that the width of the decomposition satisfies $w = \tw(\overline{G})$. We will show that there exists a tree decomposition on the graph $\widehat{G}_{\rho}$ with width at most $\rho (w + 1) - 1$, which in turn will imply that $\tw(\widehat{G}_{\rho}) \leq \rho ( w + 1) -1 = \rho (\tw(\overline{G}) +1) - 1$. Indeed, consider the following sets:
\begin{equation*}
    \widehat{X}_j := \{ u_i : u \in X_j, i \in [\rho] \},
\end{equation*}
for all $j \in [k]$. In words, each node $V(\overG) \ni u \in X_j$ is replaced by all of its copies $u_i$ in $\widehat{X}_j$. Observe that, by construction, $|\widehat{X}_j| = \rho |X_j|$. Thus, it suffices to show that the collection of sets $\{\widehat{X}_j \}_{j=1}^k$ forms a legitimate tree decomposition. First, since $V(\overline{G}) \subseteq \bigcup_j X_j$, it follows that $V(\widehat{G}_{\rho}) \subseteq \bigcup \widehat{X}_j$. Moreover, consider any two sets $\widehat{X}_j, \widehat{X}_{\ell}$, both containing a node $u_i \in V(\widehat{G}_{\rho})$ for some $i \in [\rho]$. Then, we know that all the tree-nodes in the (unique) path between $X_j$ and $X_{\ell}$ based on the original tree decomposition include $u$ since $X_j$ and $X_{\ell}$ both include $u$ and $\{ X_j \}$ is a tree decomposition of $\overline{G}$. In turn, this implies that all the tree-nodes in the path between $\widehat{X}_j$ and $\widehat{X}_{\ell}$ also contain $u_i$. Thus, the tree-nodes containing $u_i$ form a connected subtree. Finally, we know that for every edge $\{u, v\} \in E(\overline{G})$ there exists a subset $X_j$ such that $u, v \in X_j$. Hence, we can infer that for every edge in $E(\widehat{G}_{\rho})$ there is a tree-node $\widehat{X}_j$ which includes both incident endpoints. As a result, we have constructed a tree decomposition in $\widehat{G}_{\rho}$ with width $\max_{j \in [k]} |\widehat{X}_j| - 1 \leq \rho (w + 1) - 1$.
\end{proof}

\begin{corollary}
    \label{corollary:treewidth}
  Let $\overG$ be an $\overn$-node communication network of diameter at most $D$ and treewidth $\tw(\overline{G})$. Then, we can solve with high probability any $\rho$-congested part-wise aggregation problem in $\overG$ within $\widetilde{O}(\rho^2 \cdot \tw(\overline{G}) \cdot D )$ rounds of $\congest$.
\end{corollary}
\begin{proof}
    First, we know from \Cref{lemma:simulated_density} that $\tw(\hatG_{\rho}) = O(\rho \tw(\overG))$, in turn implying that the minor density of $\widehat{G}_{\rho}$ can be bounded as $\delta(\hatG_{\rho}) \le \tw(\hatG_{\rho}) = O(\rho \tw(\overG))$ (\Cref{fact:minor_density}). Thus, \Cref{theorem:quality-existence} implies that $\widehat{G}_{\rho}$ admits shortcuts of quality $\widetilde{O}(\rho \tw(\overline{G}) D(\overline{G}))$, which can be additionally constructed in $\widetilde{O}(\rho \tw(\overline{G}) D(\overline{G}))$ rounds of communication on $\widehat{G}_{\rho}$. 
    Finally, we have shown in \Cref{pa-layered-to-pa-congested} that this is sufficient to solve any $\rho$-congested part-wise aggregation problem on $\overline{G}$ in $\widetilde{O}(\rho^2 \cdot tw(\overG) \cdot D(\overline{G}))$ rounds of $\congest$, concluding the proof.
\end{proof}

\paragraph{Minor Density in the Layered Graph} In light of \Cref{lemma:simulated_density}, a natural question is whether an analogous bound holds with respect to the minor density of the underlying graph; i.e., whether $\delta(\widehat{G}_{\rho}) = \poly(\rho) \delta({G})$. Such a result would be strictly stronger as it would apply to the broader class of graphs with bounded minor density, and would essentially lift all the results in~\citep{DBLP:journals/corr/abs-2008-03091}, such as \Cref{theorem:quality-existence}, to the node-congestion setting in a black-box manner. Unfortunately, this is not possible. 

Indeed, consider a $\sqrt{{n}} \times \sqrt{{n}}$ grid ${G}$---where $\sqrt{{n}}$ is assumed to be an integer---such that every node in the graph is $2$-congested. Then, it is clear that $\delta({G}) = \widetilde{O}(1)$ (since planar graphs have excluded minors). On the other hand, we claim that $\delta(\widehat{G}_{\rho}) = \Omega(\sqrt{{n}})$. To see this, denote by $(i, j)$ the node positioned in the $i$-th row and $j$-th column with respect to the original graph, and by $(i', j')$ the node positioned in the $i$-th row and $j$-th column of the "duplicate" layer, for $i, j \in [\sqrt{{n}}]$. Moreover, let $C_j = \{ (i, j) : i \in [\sqrt{{n}}] \}$ be the nodes comprising the $j$-th column of the original graph and $R_{i} = \{ (i', j') : j' \in [\sqrt{{n}}] \}$ be the nodes comprising the $i$-th row of the duplicate layer. Then, it follows that the minor graph induced by the connected components $R_1, \dots, R_{\sqrt{{n}}}, C_1, \dots, C_{\sqrt{{n}}}$ contains the complete bipartite graph $K_{\sqrt{{n}}, \sqrt{{n}}}$ as a subgraph (\Cref{fig:minordensity}). As a result, this implies that the minor density of $\widehat{G}_{\rho}$ is $\Omega(\sqrt{{n}})$.

\begin{observation}
    \label{observation:minordensity}
    There exists an $n$-node graph ${G}$ with minor density $\delta({G}) = \widetilde{O}(1)$, but its $2$-layered version $\widehat{G}_{2}$ has minor density $\delta(\widehat{G}_{2}) = \Omega(\sqrt{{n}})$.
\end{observation}

\begin{figure}[!ht]
  \centering
  \includegraphics[scale=0.65]{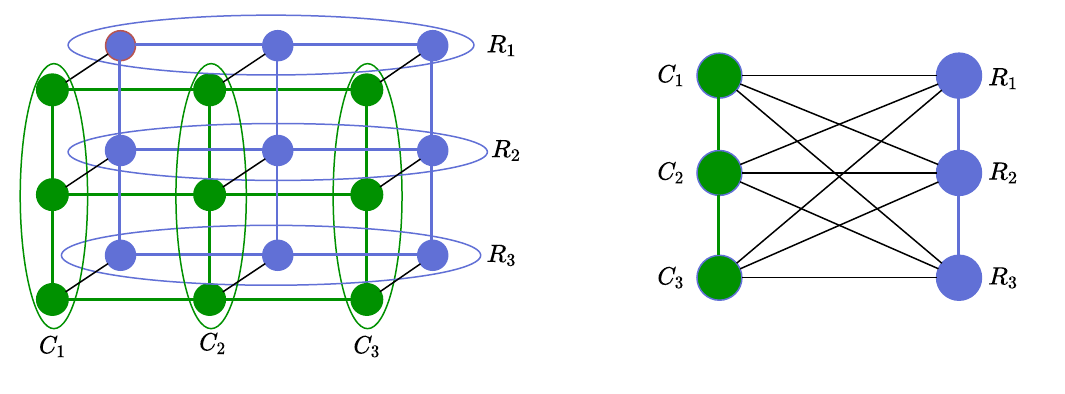}
  \caption{The layered graph $\widehat{G}_{\rho}$ corresponding to a $3 \times 3$ grid with every node having congestion $\rho = 2$ (leftmost image), and a minor of $\widehat{G}_{\rho}$ induced by the connected components $\{C_1, C_2, C_3, R_1, R_2, R_3\}$ (rightmost image).}
  \label{fig:minordensity}
\end{figure}

\subsubsection{General Graphs}
\label{sec:general-partwise}

We conclude with our main result of \Cref{subsection:congest}: a near-optimal distributed algorithm for solving the $\rho$-congested part-wise aggregation problem in general graphs. In light of our reduction in \Cref{pa-layered-to-pa-congested}, the technical crux is to control the degradation in the shortcut quality incurred by the transformation into the layered graph. Surprisingly, we show that the shortcut quality of $\hatG_{\rho}$ does not increase by more than a polylogarithmic factor even when the number of layers is polynomial:

\begin{restatable}{theorem}{theoremGrhoSmallShortcutQuality}\label{theorem:Grho-small-shortcut-quality}
  For any $\overn$-node graph $\overG$ and any $\mathbb{Z}_{\ge 1} \ni \rho \le \poly(\overn)$, we have that $\SQ(\hatG_{\rho}) = \tl{O}(\SQ(\overG))$.
\end{restatable}

This theorem improves over our previous result for treewidth-bounded graphs (\cref{lemma:simulated_density}) since the latter guarantee inevitably induces a linear factor of $\rho$ in the shortcut quality of $\hatG_{\rho}$. While this will not affect the asymptotic performance of the Laplacian solver, this improvement might prove to be important for future applications. Assuming that we have shown \Cref{theorem:Grho-small-shortcut-quality}, we can then utilize the efficient shortcut constructions given in \Cref{thm:big-boy-construction} to solve $\rho$-congested part-wise aggregations on any graph. 

\begin{restatable}{corollary}{congestedPaGeneralSolving}\label{congested-pa-general-solving}
  There exists a randomized distributed algorithm that, for any $\overn$-node graph $\overG$ and $\rho \in \mathbb{Z}_{\ge 1} \leq \poly(\overn)$, solves with high probability any $\rho$-congested part-wise aggregation instance on $\overG$ with the following guarantees:
  \begin{itemize}\setlength\itemsep{0em}
  \item In the $\congest$, the algorithm terminates in at most $\rho \cdot \poly\!\left(\SQ(\overG) \right) \cdot \overline{n}^{o(1)}$ rounds.
  \item In the $\congest$ model on graphs with minor density $\delta$, it requires $\tl{O}(\rho \cdot \delta \cdot D)$ rounds.
  \item In the Supported-$\congest$, the algorithm terminates in $\tl{O}(\rho \cdot \SQ(\overG))$ rounds.
  \end{itemize}
\end{restatable}
\begin{proof}
  We construct (virtually) the graph $\hatG_{O(\rho)}$. By \Cref{theorem:Grho-small-shortcut-quality}, we know that $\SQ(\hatG_{O(\rho)}) = \tl{O}(\SQ(\overG))$. Thus, we can construct shortcuts on $\hatG_{O(\rho)}$ in $\poly\!\left(\SQ(\overG)\right) \cdot \overline{n}^{o(1)}$, $\widetilde{O}(\delta D)$, and $\widetilde{O}(\SQ(\overG))$ rounds in general CONGEST (\Cref{thm:big-boy-construction}), Supported-CONGEST (\Cref{thm:big-boy-construction}), and CONGEST with minor density $\delta$ (\Cref{theorem:quality-existence}), respectively. Therefore, we can solve $1$-congested part-wise aggregation instances using those shortcuts in the required times using \Cref{shortcuts-solving-pa}. Since solving a $1$-congested part-wise aggregation on $\hatG_{O(\rho)}$ suffices to solve $\rho$-congested part-wise aggregations on $\overG$ with only an $\tl{O}(\rho)$ slowdown (\Cref{pa-layered-to-pa-congested}), the proof is complete.
\end{proof}

The rest of this subsection is dedicated to the proof of \Cref{theorem:Grho-small-shortcut-quality}. To argue about the shortcut quality of the layered graph, we need to develop several generalized notions of node connectivity. Pair node and any-to-any connectivity are essentially the multi- and single-commodity versions of node connectivity, respectively.

\paragraph{Pair Node Connectivity} Given a (multi)set of source-sink pairs $\calP = \{ (s_i, t_i) \}_{i=1}^k$ in $G$, we say that $\calP$ has pair node connectivity $\rho$ if there exist paths $P_1, \ldots, P_k$, with $s_i$ and $t_i$ being the endpoints of each $P_i$, such that every node $v \in V(G)$ is contained in at most $\rho$ many paths, i.e., for all $v$ we have $|\{ i : V(P_i) \ni v \}| \le \rho$. If $\calP$ has pair node connectivity $1$ we say that they are pair \emph{node-disjointly connectable}.

\paragraph{Any-to-Any Node Connectivity} Suppose that we are given multisets of $k$ \emph{sources} $S = \{ s_1, \ldots, s_k \}$ and $k$ \emph{sinks} $T = \{ t_1 \ldots, t_k \}$. We say that $(S, T)$ have any-to-any node connectivity $\rho$ if there is a permutation $\pi : \{ 1, \ldots, k \} \to \{ 1, \ldots, k \}$ such that the pairs $\{ (s_i, t_{\pi(i)}) \}_{i=1}^k$ have pair node connectivity $\rho$. If $(S, T)$ have any-to-any node connectivity $1$ we say they are any-to-any \emph{node-disjointly connectable}.

The following decomposition lemma states that two sets with any-to-any node connectivity $\rho$ can be decomposed into $\tl{O}(\rho)$ many pairs of subsets that are any-to-any node-disjointly connectable.

\begin{lemma}\label{lemma:partition-k-connectable-into-disjoint}
  Given a graph $G$, suppose we are given any two multisets of nodes $S \subseteq V(G)$ and $T \subseteq V(G)$ of size $k := |S| = |T|$ that have any-to-any node connectivity $\rho$. Then, we can partition $S = S_1 \uplus S_2 \uplus \ldots \uplus S_{O(\rho \log k)}$ and $T = T_1 \uplus T_2 \uplus \ldots T_{O(\rho \log k)}$ such that $|S_i| = |T_i|$ and $(S_i, T_i)$ are any-to-any node-disjointly connectable.
\end{lemma}
\begin{proof}
  Suppose that each edge in $G$ has infinite capacity while each node in $G$ has unit capacity. Then, let us connect a super-source $s$ to each node $x \in S$ with a unit-capacity edge, and a super-sink $t$ to each node $x \in T$ with a unit capacity edge. By assumption, we know that there exists a flow $f$ over $E(G)$ which sends $k$ units of flow from $s$ to $t$ with edge congestion $1$ and node congestion at most $\rho$. Therefore, the flow $f / \rho$ sending $k / \rho$ units of flow from $s$ to $t$ is a feasible solution of the maximum flow linear program with node constraints (i.e., it satisfies both edge and node capacity constraints). Since that linear program is integral (i.e., has an integrality gap of $1$), there exists an integral flow $f'$ which sends at least $k / \rho$ units of flow and satisfies both node and edge capacity restrictions. In other words, there exist at least $k / \rho$ node disjoint paths (with the exception of the endpoints) between $s$ and $t$. Let $S_1 \subseteq S$ ($T_1 \subseteq T$) be the set of nodes on these paths immediately following the super-source (just before the super-sink, respectively). Clearly, by construction, $(S_1, T_1)$ are any-to-any node-disjointly connectable. Finally, we define $S' := S \setminus S_1, T' := T \setminus T_1$ and proceed iteratively as above (producing $S_2, T_2$ instead of $S_1, T_1$). In each step, the size of $S'$ and $T'$ decreases by at least a multiplicative factor of $1 - 1/\rho$. Hence, $O(\rho \log k)$ steps suffice so that $S' = T' = \emptyset$. 
\end{proof}  


Next, we introduce two communication tasks that will be useful for characterizing the shortcut quality.
\paragraph{Multiple-Unicast Problem} Suppose that we are given $k$ source-sink pairs $\calP = \{ (s_i, t_i) \}_{i=1}^k$. The goal is to find the smallest possible \emph{completion time} $\tau$ such that there are $k$ paths $P_1, \ldots, P_k$ for which (1) the endpoints of each $P_i$ are exactly $s_i$ and $t_i$; (2) the dilation is $\tau$, i.e., each path $P_i$ has at most $\tau$ hops; and (3) the congestion is $\tau$, i.e., each edge $e \in E(G)$ is contained in at most $\tau$ many paths.

\paragraph{Any-to-Any-Cast Problem} Suppose we are given $k$ \emph{sources} $S = \{ s_1, \ldots, s_k \}$ and $k$ \emph{sinks} $T = \{ t_1 \ldots, t_k \}$. The goal is to find the smallest \emph{completion time} $\tau$ such that there exists a permutation $\pi : \{ 1, \ldots, k \} \to \{ 1, \ldots, k \}$ for which the multiple-unicast problem on $\{ ( s_i, t_{\pi(i)} ) \}_{i=1}^k$ has a completion time of at most $\tau$.


Finally, we now recall (a reinterpretation of) a result characterizing shortcut quality from ~\cite{haeupler2020network,DBLP:conf/stoc/HaeuplerWZ21}. Shortcut quality was originally defined as the smallest completion-time of the worst-case generalized (with respect to parts) multiple-unicast (i.e., multi-commodity) problem over an a \emph{pair} node-disjointly connectable instance (\Cref{def:shortcut-quality}). Using recent network coding gap results we can equivalently express shortcut quality as the smallest completion-time of the worst-case any-to-any-cast (i.e., single-commodity) problem over sources and sinks that are \emph{any-to-any} node-disjointly connectable. The formal statement follows.
\begin{theorem}[\citep{haeupler2020network,DBLP:conf/stoc/HaeuplerWZ21}] \label{theorem:sq-vs-worst-any-any-cast}
  Consider any graph $G$ and let $\tau$ be the worst-case completion time of any-to-any-cast problems taken over all any-to-any node-disjointly connectable sets $(S \subseteq V(G), T \subseteq V(G))$. Then, $\tau = \tl{\Theta}(\SQ(G))$.
\end{theorem}
\begin{proof}
  It was proven in \cite[Lemma 2.8 in the Full Version]{DBLP:conf/stoc/HaeuplerWZ21} that $\SQ(G)$ is, up to $\tl{\Theta}(1)$ factors, equal to the completion time $C$ of some multiple-unicast instance with respect to some source-sink pairs $\calP := \{ (s_i, t_i) \}_{i=1}^k$ that are pair node-disjointly connectable. We note that, since sources and sinks are disjoint, it follows that $k = \poly(n)$ and $O(\log n) = O(\log k)$. Furthermore, \citep{haeupler2020network} proved that there exists a sub-instance $\calP' = \{ (s'_i, t'_i)_{i=1}^{k'} \} \subseteq \calP$ such that $\SQ(G)$ is (up to $\tl{\Theta}(1)$ factors) equal to the completion time $\tau$ of the any-to-any-cast problem with respect to $(\{ s'_i \}_{i=1}^{k'}, \{ t'_i \}_{i=1}^{k'})$. One side of the claim is clear: for any sub-instance $\calP' \subseteq \calP$ we have that $\tau \le C$. The other direction is harder and we sketch its proof here using the terminology in~\citep{haeupler2020network}. By definition and strong duality, $\mathrm{Cut}_{\calP}(2C) = \mathrm{ConcurrentFlow}_{\calP}(2C) \le 1$. Furthermore, $\mathrm{Cut}_{\calP}(C/10) = \mathrm{Cut}_{\calP}(2C)/20 \le 1/10$. Hence, by \citep[Lemma 2.6]{DBLP:conf/stoc/HaeuplerWZ21} there is a sub-instance $\calP' \subseteq \calP$ with a moving cut of distance $\tau := \tl{\Omega}(C)$ and capacity less than $|\calP'|$. Therefore, 
  this proves that the completion time of any-to-any-cast problem on $\calP'$ is at least $\tau$. With this in mind, we have that $\tl{\Omega}(\SQ(G)) = \tl{\Omega}(C) = \tau \le C = \tl{\Theta}(SQ(G))$.

  Finally, since $\calP = \{ (s_i, t_i) \}_{i=1}^k$ was pair node-disjointly connectable, it follows from definition that the sub-instance $(\{ s'_i \}_{i=1}^{k'}, \{ t'_i \}_{i=1}^{k'})$ is any-to-any node-disjointly connectable. Therefore, $(\{ s'_i \}_{i=1}^{k'}, \{ t'_i \}_{i=1}^{k'})$ satisfies the constraints of this result and has completion-time $\tau = \tl{\Theta}(\SQ(G))$, as required. It is also clear that, by shortcut quality, any any-to-any node-disjointly connectable instance has completion time at most $\SQ(G)$ using the node-disjoint paths that witness the any-to-any node-disjointness as parts of the shortcut, making $(\{ s'_i \}_{i=1}^{k'}, \{ t'_i \}_{i=1}^{k'})$ the worst-case such instance (modulo polylogarithmic factors).
\end{proof}

We now combine all of the previous ingredients to prove the main result of this section.

\begin{proof}[Proof of \Cref{theorem:Grho-small-shortcut-quality}]
  Let $S \subseteq V(\hatG_{\rho})$ and $T \subseteq V(\hatG_{\rho})$ be any-to-any node-disjointly connectable sets such that the completion time of any-to-any-cast between $S$ and $T$ is $\tl{\Theta}(\SQ(\hatG_{\rho}))$ (\Cref{theorem:sq-vs-worst-any-any-cast}). Let $k := |S| = |T|$, and suppose that $S' := \biguplus_{s \in S} \{ \pi(s) \} \subseteq V(\overG)$ and $T' := \biguplus_{t \in T} \{ \pi(t) \}  \subseteq V(\overG)$ are the multisets induced by projecting $S$ and $T$ to $\overG$, respectively. By construction of $\hatG_{\rho}$, $S'$ and $T'$ have any-to-any node connectivity $\rho$; to see this, consider the witness paths disjointly connecting them in $\hatG_{\rho}$ and project them to $\overG$. Therefore, we can partition $S' = S'_1 \uplus \ldots \uplus S'_{O(\rho \log k)}$ and $T' = T'_1 \uplus \ldots \uplus T'_{O(\log k)}$ such that $|S'_i| = |T'_i|$ and $(S'_i, T'_i)$ are any-to-any node-disjointly connectable in $\overG$ (\Cref{lemma:partition-k-connectable-into-disjoint}).

  By definition of shortcut quality, for each $i \in \{1, \ldots, O(\rho \log k)\}$ there exists a set of paths $(P^i_j)_{j=1}^{|S'_i|}$ in $\overG$ between $S'_i$ and $T'_i$ of quality (i.e., both congestion and dilation) at most $\SQ(\overG)$. Then, we inject the first $O(\log k)$ collections of paths $(P^1_j)_j, (P^2_j)_j, \ldots, (P_j^{O(\log k)})_j$ to the first layer $\overG_1$ of $\hatG_{\rho}$; the second $O(\log k)$ collections to the second layer $\overG^2$, and so on, until we finally inject the last $O(\log k)$ collections to the last layer $\overG_\rho$. Note that only the paths on the same layer interact, so both the congestion and dilation after injecting all paths into $\hatG_{\rho}$ is $O(\SQ(\overG) \log k)$. Hence, the same applies for the shortcut quality. Finally, to solve the any-to-any-cast problem on $S$ and $T$ one might need to add an between-layer edge at the beginning and at the end since each injected path is restricted to some adversarially chosen layer. However, this only increases the congestion and dilation by $O(1)$. Hence, the completion time of any-to-any-cast between $S$ and $T$ is $\tl{O}(\SQ(\overG))$, implying that $\SQ(\hatG_{\rho}) = \tl{O}(\SQ(\overG))$.
\end{proof}

\subsection{The \texorpdfstring{$\ncc$}{NCC} Model}
\label{subsection:ncc}

We next turn our attention to the $\ncc$ model. We observe that the $\rho$-congested part-wise aggregation problem admits a solution in $\poly(\rho, \log \overline{n})$ rounds of $\ncc$. This is established after appropriately translating the communication primitives established for $\ncc$ in \cite{DBLP:conf/spaa/AugustineGGHSKL19}; the details are provided in \Cref{appendix:ncc}.

\begin{restatable}{lemma}{lemmaCongPaNcc}
    \label{lemma:conge_PA-ncc}
    Let $\overline{G}$ be an $\overline{n}$-node communication network. Then, we can solve with high probability any $\rho$-congested part-wise aggregation problem on $\overline{G}$ after $O(\rho + \log \overline{n})$ rounds of $\ncc$.
\end{restatable}

\section{Almost Universally Optimal Laplacian Solvers}
\label{section:Laplacian}

In this section we relate the congested part-wise aggregation problem we studied in the previous section with the Laplacian solver of \citet{DBLP:journals/corr/abs-2012-15675}. To present a unifying analysis for both $\congest$ and $\hybrid$, as well as for future applications and extensions, we analyze the distributed Laplacian solver under the following hypothesis.

\begin{assumption}
    \label{assumption:minor_aggregation}
    Consider a model of computation which incorporates $\congest$. We assume that we can solve with high probability any $\rho$-congested part-wise aggregation problem in $Q(\rho) = O(\rho^c Q(1))$ rounds, for some universal constant $c \geq 1$.
\end{assumption}

One of our crucial observations is that the performance of the Laplacian solver of \citet{DBLP:journals/corr/abs-2012-15675} can be parameterized in terms of the complexity of the congested part-wise aggregation problem. Indeed, we revisit and refine the main building blocks of their solver in \Cref{section:solver}, leading to the following result.

\begin{theorem}[Full Version in \Cref{theorem:laplacian-abstract-full}]
    \label{theorem:laplacian-abstract}
Consider a weighted $\overline{n}$-node graph $\overG$ for which \Cref{assumption:minor_aggregation} holds for some $Q(\rho) = O(\rho^c \mathcal{Q})$, where $c$ is a universal constant and $\mathcal{Q} = \mathcal{Q}(\overG)$ is some parameter. Then, we can solve any Laplacian system after $\overline{n}^{o(1)} \mathcal{Q} \log(1/\eps)$ rounds.
\end{theorem}

Combining this theorem with \Cref{congested-pa-general-solving} and \Cref{lemma:conge_PA-ncc} yields the following immediate consequences.
%
%

\thmLaplacianFullModels*
\thmLaplacianHybrid*




\paragraph{Lower Bound in Supported-\texorpdfstring{$\congest$}{CONGEST}}

Finally, we complement our positive results with a almost-matching lower bound on any graph $\overline{G}$, applicable even under the Supported-$\congest$ model, thereby establishing universal optimality up to an $\overline{n}^{o(1)}$ factor. Our reduction leverages the refined hardness result established in~\citep{DBLP:conf/stoc/HaeuplerWZ21} for the \emph{spanning connected subgraph} problem~\citep{10.1145/1993636.1993686}. In this problem a subgraph $\overline{H}$ of $\overline{G}$ is specified with nodes knowing all of the incident edges belonging to $\overline{H}$. The goal is to let every node learn whether $\overline{H}$ is connected and spans the entire network.

\begin{theorem}[\citep{DBLP:conf/stoc/HaeuplerWZ21}]
    \label{theorem:shortcutquality-lb}
  Let $\mathcal{A}$ be any algorithm which is always correct with probability\footnote{Note that \citet{DBLP:conf/stoc/HaeuplerWZ21} only proved this for always-correct algorithms with probability $1$, but the extension we claim here follows readily from their argument.} at least $\frac{2}{3}$ for the spanning connected subgraph problem, and $T(\overline{G}) = \max_{\mathcal{I}}T_{\mathcal{A}}(\mathcal{I}; \overline{G})$ be the worst-case round-complexity of $\mathcal{A}$ under $\overline{G}$. Then, 
  \begin{equation*}
      T(\overline{G}) = \widetilde{\Omega}(\textsc{ShortcutQuality}(\overline{G})).
  \end{equation*}
\end{theorem}

In this context, we show that a Laplacian solver can be leveraged to solve the spanning connected subgraph problem, leading to the following lower bound.






\lowerb*

This substantially strengthens the \emph{existential lower bound} in~\citep{DBLP:journals/corr/abs-2012-15675}, and deviates from their argument which is based on a reduction from the $s-t$ connectivity problem. The proof is deferred to \Cref{appendix:lb}.

\section{Conclusions}

We established almost universally optimal Laplacian solvers for both the (Supported-)$\congest$ and the $\hybrid$ model. One of our main technical contributions was to introduce and study a congested generalization of the standard part-wise aggregation problem, which we believe may find further applications beyond the Laplacian paradigm in the future. For example, one candidate problem would be to refine the distributed algorithm for max-flow due to \citet{10.1145/2767386.2767440}. We also hope that our accelerated Laplacian solvers will be used as a basic primitive for obtaining improved distributed algorithms for other fundamental optimization problems as well. Indeed, \citet{DBLP:journals/corr/abs-2012-15675} showed that the Laplacian paradigm can offer \emph{sublinear} and \emph{exact} distributed algorithms for problems such as max-flow, an objective which previously appeared elusive.

\printbibliography



\clearpage

\appendix

\section{The Laplacian Solver}
\label{section:solver}

In this section we describe the basic building blocks of the distributed Laplacian solver of~\citep{DBLP:journals/corr/abs-2012-15675}. Our goal will be to cast their guarantees within our more general framework, leading to the proof of \Cref{theorem:laplacian-abstract}. First, let us introduce some notation related to Laplacian systems.

\paragraph{The Laplacian Matrix} Consider a weighted undirected graph $G = (V, E, \vec{w} > 0)$. The \emph{Laplacian} of the graph $G$ is defined as 
\begin{equation*}
    \mathcal{L}(G)_{u, v} =
\begin{cases}
\sum_{\{u, z \} \in E} \vec{w}(u, z) & \textrm{If $u = v$}, \\
- \vec{w}(u, v) & \textrm{otherwise}.
\end{cases}
\end{equation*}

The Laplacian matrix of a graph is (i) \emph{symmetric} ($\mathcal{L}(G)^T = \mathcal{L}(G)$); (ii) \emph{positive semi-definite} ($\vec{x}^T \mathcal{L}(G) \vec{x} \geq 0$ for any $\vec{x}$); and (iii) \emph{weakly diagonally dominant} ($\mathcal{L}(G)_{u, u} \geq \sum_{v \neq u} |\mathcal{L}(G)_{u, v}|$).

\begin{definition}[Schur Complement]
    For a symmetric matrix $\mat{A} \in \mathbb{R}^{n \times n}$ and a partition of $[n]$ into $\mathcal{T}$ and $S$, permute the rows and columns of $\mat{A}$ such that 
    \begin{equation*}
        \mat{A} = \begin{bmatrix}
            \mat{A}_{[S, S]} & \mat{A}_{[S, \mathcal{T}]} \\
            \mat{A}_{[\mathcal{T}, S]} & \mat{A}_{[\mathcal{T}, \mathcal{T}]}
        \end{bmatrix}.
    \end{equation*}
    Then, the \emph{Schur compelement} of $\mat{A}$ onto $\mathcal{T}$ is defined as $\schur(\mat{A}, \mathcal{T}) := \mat{A}_{[\mathcal{T}, \mathcal{T}]} - \mat{A}_{[\mathcal{T}, S]} \mat{A}^{\dagger}_{[S, S]} \mat{A}_{[S, \mathcal{T}]}$, where $\mat{M}^{\dagger}$ denotes the \emph{Moore-Penrose pseudo-inverse} of matrix $\mat{M}$. For a graph $G$ and a subset $\mathcal{T} \subseteq V(G)$, we will write $\schur(G, \mathcal{T}) := \schur(\mathcal{L}(G), \mathcal{T})$.
\end{definition}
 
\paragraph{Notation} Consider two positive semi-definite matrices $\mat{A}, \mat{B} \in \mathbb{R}^{n \times n}$. For a vector $\vec{x} \in \mathbb{R}^n$ we define $\|\vec{x}\|_{\mat{A}} := \sqrt{\vec{x}^T \mat{A} \vec{x}}$ (Mahalanobis norm). We will write $\mat{A} \approxeps \mat{B}$ if $\exp(-\eps) \mat{A} \preceq \mat{B} \preceq \exp (\eps) \mat{A}$, where $\mat{A} \preceq \mat{B}$ if and only if the matrix $\mat{B} - \mat{A}$ is positive semi-definite. For an edge $e = \{u, v\}$, we will let $\vec{b}(e) := \mathbbm{1}_{u} - \mathbbm{1}_{v}$, where $\mathbbm{1}_{u} \in \mathbb{R}^n$ represents the characteristic vector of node $u$. For a graph $G$ with \emph{resistances} $\vec{r}(e)$, we define the \emph{leverage scores} as $\lev_G(e) := \vec{r}(e)^{-1} \vec{b}^T(e) \mathcal{L}(G)^{\dagger} \vec{b}(e)$. Note that $0 \leq \lev_G(e) \leq 1$.


\subsection{Low-Congestion Minors}

Here we introduce the concept of a \emph{low-congestion minor}, a central component in the distributed Laplacian solver of \citep{DBLP:journals/corr/abs-2012-15675}. 

\begin{definition}[\cite{DBLP:journals/corr/abs-2012-15675}]
    \label{definition:low_congestion-minor}
A graph $G$ is a \emph{minor} of $\overline{G}$ if the following properties hold:
\begin{enumerate}
    \item For every node $u^G \in V(G)$ there exists: 
    \begin{itemize}
        \item[(i)] A subset of nodes of $\overline{G}$, which is termed as a \emph{super-node}, $\supervertex(u^G)$, with a leader node $\ell(u^G) \in \supervertex(u^G)$;
        \item[(ii)] A connected subgraph of $\overline{G}$ on $\supervertex(u^G)$, for which we maintain a spanning tree $\tree(u^G)$.
    \end{itemize}
    \item There exists a mapping of the edges of $G$ onto edges of $\overline{G}$, or self-loops, such that for any $\{u^G, v^G\} \in E(G)$, the mapped edge $\{u, v\}$ satisfies $u \in \supervertex(u^G)$ and $v \in \supervertex(v^G)$.
\end{enumerate}
Moreover, we say that this minor $G$ has \emph{congestion} $\rho$, or $G$ is a $\rho$-minor, if:
\begin{enumerate}
    \item Every node $u \in \overline{G}$ is contained in at most $\rho$ super-nodes $\supervertex(u^G)$, for some $u^G \in V(G)$;
    \item Every edge of $\overline{G}$ appears as the image of an edge of $G$ or in one of the trees connecting super-nodes (i.e., $\tree(u^G)$ for some $u^G$) at most $\rho$ times.
\end{enumerate}
Finally, we say that $G$ is $\rho$\emph{-minor distributed} over $\overline{G}$ if every $u \in V(\overline{G})$ stores:
\begin{enumerate}
    \item All $u^{G} \in V(G)$ for which $u \in \supervertex(u^G)$;
    \item For every edge $e$ incident to $u$, (i) all the nodes $u^G$ for which $e \in \tree(u^G)$, and (ii) all edges $e^{G}$ that map to it.
\end{enumerate}
\end{definition}

We remark that the basis of \Cref{definition:low_congestion-minor} was the earlier concept of a \emph{distributed cluster graph} of \citet{10.1145/2767386.2767440}. The important connection is that the congested part-wise aggregation problem we introduced is the central ingredient that allows performing certain ``local'' operations on a graph $\rho$-minor distributed into the underlying communication network. The following lemma is a direct consequence of \Cref{definition:low_congestion-minor}.

\begin{lemma}
    \label{lemma:minor_aggregation}
Let $G = (V, E)$ be an $n$-node graph $\rho$-minor distributed into an $\overline{n}$-node communication network $\overline{G} = (\overline{V}, \overline{E})$ for which \Cref{assumption:minor_aggregation} holds for some $Q = Q(\rho)$. Then, we can perform with high probability the following operations in the $\ncc$ model, simultaneously for all $u^G \in V(G)$, within $O(Q(\rho))$ rounds:

\begin{enumerate}
    \item Every leader $\ell(u^G)$ sends an $O(\log \overline{n})$-bit message to all the nodes in $\supervertex(u^G)$;
    \item All the nodes in $\supervertex(u^G)$ compute an aggregation function on $O(\log \overline{n})$-bit inputs.
\end{enumerate}
\end{lemma}

\subsection{The Laplacian Building Blocks}

To keep the exposition reasonably self-contained, here we review the basic ingredients of the distributed Laplacian solver developed in \citep{DBLP:journals/corr/abs-2012-15675}. Our main goal is to extend the guarantees established in \citep{DBLP:journals/corr/abs-2012-15675} under \Cref{assumption:minor_aggregation}. Then, we will combine these pieces in \Cref{sub:everything} to complete the construction.

\subsubsection{Ultra-Sparsification}

As is standard in the Laplacian paradigm, we will require a preconditioner in the form of an \emph{ultra-sparsifier}. In particular, the following lemma is established in \Cref{appendix:ultra_sparsification}, and it is a refinement of~\citep[Lemma 4.9]{DBLP:journals/corr/abs-2012-15675}:

\begin{lemma}[Ultra-Sparsification] 
    \label{lemma:ultra_sparsification}
Consider an $n$-node $m$-edge graph $G$ which is $\rho$-minor distributed into an $\overline{n}$-node communication network $\overline{G}$ for which \Cref{assumption:minor_aggregation} holds for some $Q = Q(\rho)$. Then, $\ultraspars(G, k)$ takes as input a parameter $k$ and returns after $n^{o(1)} Q(\rho)$ rounds a graph $H$ such that 

\begin{enumerate}
    \item $H$ is a subgraph of $G$;
    \item $H$ has $n - 1 + m 2^{O(\sqrt{\log n \log \log n})}/k$ edges;
    \item $\mathcal{L}(G) \preceq \mathcal{L}(H) \preceq k \mathcal{L}(G)$.
\end{enumerate}
Moreover, the algorithm returns $\widehat{G}, \mat{Z}_1, \mat{Z}_2, C$ such that 

\begin{enumerate}
    \item $\widehat{G}$ $1$-minor distributes into $H$ such that $\widehat{G} = \schur(H, C)$, with $|C| = m 2^{O(\sqrt{\log n \log \log n})} /k$;
    \item The operators $\mat{Z}_1$ and $\mat{Z}_2$ can be evaluated in $O(Q(\rho) \log n)$ rounds, and are such that 
    \begin{equation*}
        \mathcal{L}(H)^{\dagger} = \mat{Z}_1^T 
        \begin{bmatrix}
            \mat{Z}_2 & 0 \\
            0 & \mathcal{L}(\widehat{G})^{\dagger}
        \end{bmatrix}
        \mat{Z}_1.
    \end{equation*}
\end{enumerate}
\end{lemma}
Let us briefly review the pieces required for this lemma. First, we need the distributed implementation of the low-stretch spanning tree algorithm of \citet{10.1137/S0097539792224474} which is due to \citet{10.1145/2767386.2767440}. Then, this spanning tree is augmented with off-tree edges based on the sampling procedure of \citet{5671167}, leading to a graph with a spectral approximation guarantee with respect to the original graph. Finally, the parallel elimination procedure of \citet{DBLP:journals/mst/BlellochGKMPT14} is used to perform a series of contractions, leading to a subset with size analogous to the number of off-tree edges. We revisit these steps in detail in \Cref{appendix:ultra_sparsification}.

\subsubsection{Sparsified Cholesky}

The next building block is the sparsified Cholesky algorithm of \citet{10.1145/2897518.2897640}, which manages to effectively eliminate in every iteration a non-negligible fraction of the nodes. In the distributed context, we state the following lemma which is a refinement of \cite[Lemma 4.10]{DBLP:journals/corr/abs-2012-15675}.

\begin{lemma}[Sparsified Cholesky]
    \label{lemma:eliminate}
    Let $G$ be an $n$-node graph $\rho$-minor distributed into a communication network $\overline{G}$ for which \Cref{assumption:minor_aggregation} holds for some $Q = Q(\rho)$. Then, for a given parameter $d$ and error $\eps$, the algorithm $\eliminate(G, d, \eps)$ runs in $O(Q(\rho)(\log^{c} n /\eps^c)^d)$ rounds, where $c$ represents some universal constant, and returns a subset $\mathcal{T} \subset V(G)$ and access to operators $\mat{Z}_1$ and $\mat{Z}_2$ such that 
    \begin{enumerate}
        \item $|\mathcal{T}| \leq (49/50)^d |V(G)|$;
        \item The operators $\mat{Z}_1, \mat{Z}_1^T, \mat{Z}_2$ can be applied to vectors in $O( Q(\rho)(\log^{c} n /\eps^c)^d)$ rounds;
        \item 
        \begin{equation*}
            (1 - \eps)^d \mathcal{L}(G)^{\dagger} \preceq \mat{Z}_1^T 
            \begin{bmatrix}
                \mat{Z}_2 & 0 \\
                0 & \schur(G, \mathcal{T})^{\dagger}
            \end{bmatrix}
            \mat{Z}_1 \preceq (1 + \eps)^d \mathcal{L}(G)^{\dagger}.
        \end{equation*}
    \end{enumerate}
\end{lemma}

This lemma is established based on a distributed implementation of the \emph{sparsified Cholesky} algorithm of \citet{10.1145/2897518.2897640}. In particular, the Cholesky decomposition essentially reduces solving a Laplacian to inverting (i) any sub-matrix of the Laplacian induced on a set $S$, and (ii) the Schur complement on $V \setminus S$. Thus, \citet{10.1145/2897518.2897640} initially develop a procedure for identifying an "almost independent" subset of nodes $F$ (more precisely, a \emph{strongly} diagonally dominant subset) for which inverting the Laplacian restricted on $F$ can be done efficiently through preconditioning (e.g. via the \emph{Jacobi method}), while $F$ also contains at least a constant fraction of the nodes. Next, a combinatorial view of the Schur complement based on a certain family of random walks (see \citep{durfee2019fully}) is employed to construct a spectral sparsifier of the Schur complement on $\mathcal{T} = V \setminus F$. This process is then repeated for $d$ iterations, leading to \Cref{lemma:eliminate}. Several technical challenges that arise are discussed in \Cref{appendix:eliminate}. Next, the main idea is to recurse on the set of terminals $\mathcal{T}$. However, in our context this requires maintaining the invariant that the underlying subgraph is cast as a minor (with a reasonable congestion) of $\overline{G}$. This is ensured in the following subsection.

\subsubsection{Minor Schur Complement}

This subsection introduces a subroutine that will be invoked after the $\eliminate$ algorithm to return a low-congestion minor based on the set of terminals $\mathcal{T}$ returned by $\eliminate$; while doing so, the algorithm will incur a small overhead in the spectral guarantee, and a limited growth in the number of nodes with respect to $\mathcal{T}$. This increase will be eventually negligible due to the selection of parameter $d$ in $\eliminate$. In this context, the following lemma is a refinement of \cite[Theorem 3]{DBLP:journals/corr/abs-2012-15675}.

\begin{lemma}
    \label{lemma:approxSC}
    Let $G$ be an $n$-node graph $\rho$-minor distributed into an $\overline{n}$-node communication network for which \Cref{assumption:minor_aggregation} holds for some $Q = Q(\rho)$. Then, for an error parameter $0 < \eps < 0.1$ and a subset $\mathcal{T}$ of nodes, the algorithm $\approxSC$ returns with high probability a graph $H$ as a $\rho$-minor distribution into $\overline{G}$ such that 
    
    \begin{enumerate}
        \item $\mathcal{T} \subseteq V(H)$;
        \item $H$ has $O(|\mathcal{T}| \log^2 n/\eps^2)$ edges;
        \item $\schur(H, \mathcal{T}) \approxeps \schur(G, \mathcal{T})$.
    \end{enumerate}
    This algorithm requires $O(\log^{10} n/\eps^3)$ calls to a distributed Laplacian solver to accuracy $1/\poly(n)$ on graphs that $2\rho$-minor distribute into $\overline{G}$, and an overhead of $O(Q(\rho) \log^{10} \overline{n}/\eps^3)$ rounds.
\end{lemma}

This result builds upon the work of \citet{DBLP:conf/focs/LiS18}, who (roughly speaking) established that randomly contracting an edge with probability equal to its leverage score (and otherwise deleting) would suffice. In the distributed context, \citet{DBLP:journals/corr/abs-2012-15675} devise a parallelized implementation of this scheme based on the \emph{localization of electrical flows}~\citep{10.5555/3174304.3175408}. More precisely, they manage to identify a non-negligible subset of edges---which they refer to as \emph{steady edges}---with small mutual (electrical) ``correlation'', allowing for independent (and hence highly parallelized) contractions/deletions within this set. This approach employs the recursive and sketching-based method of random projections due to \citet{DBLP:conf/stoc/SpielmanS08}, similarly to~\citep{DBLP:conf/focs/LiS18}, to estimate quantities such as leverage scores and electrical correlation. These steps are carefully reviewed in \Cref{appendix:approxSC}.

\subsubsection{Schur Complement Chain}

Finally, let us introduce the concept of a Schur complement chain, and explain how it can be employed to produce a Laplacian solver.

\begin{definition} 
    \label{definition:schur_chain}
    For an $n$-node graph $G$, $\{(G_i, \mat{Z}_{i, 1}, \mat{Z}_{i, 2}, \mathcal{T}_i) \}_{i=1}^t$ is a $(\gamma,\eps)$-Schur complement chain if the following conditions hold:
    \begin{enumerate}
        \item $G_1 = G$;
        \item $\mathcal{T}_i \subset V(G_{i+1}) \subset V(G_i)$ and $\schur(G_i, \mathcal{T}_i) \approxeps \schur(G_{i+1}, \mathcal{T}_i)$;
        \item $|V(G_{i+1})| \leq |V(G_i)|/\gamma$ for $i < t$, and $|V(G_t)| \leq \gamma$.
        \item 
        \begin{equation*}
            (1 - \eps) \mathcal{L}(G_i)^{\dagger} \preceq \mat{Z}_{i, 1}^T 
            \begin{bmatrix}
                \mat{Z}_{i, 2} & 0 \\
                0 & \schur(G_i, \mathcal{T}_i)^{\dagger}
            \end{bmatrix}
            \mat{Z}_{i, 1} \preceq (1 + \eps) \mathcal{L}(G_i)^{\dagger}.
        \end{equation*}
    \end{enumerate}
\end{definition}

In the sequel, a Schur complement chain will be developed through \Cref{lemma:ultra_sparsification,lemma:eliminate,lemma:approxSC}. Next, the following lemma implies a solution to the Laplacian system based on a suitable Schur complement chain.

\begin{lemma}[\cite{DBLP:journals/corr/abs-2012-15675}]
    \label{lemma:schur_chain}
    Consider an $\overline{n}$-node communication network for which \Cref{assumption:minor_aggregation} holds for some $Q = Q(\rho)$, and let $\{ (G_i, \mat{Z}_{i, 1}, \mat{Z}_{i, 2}, \mathcal{T}_i ) \}_{i=1}^t$ be a $(\gamma, \eps)$-Schur complement chain for an $n$-node graph $G$ for some $\gamma \geq 2$ and $\eps \leq 1/(C \log n)$, for a sufficiently large constant $C$, such that for all $i$:
    \begin{enumerate}
        \item $G_i$ $\rho$-minor distributes into $\overline{G}$;
        \item The linear operators $\mat{Z}_{i, 1}$ and $\mat{Z}_{i, 2}$ can be evaluated in at most $\overline{n}^{o(1)} Q(\rho)$ rounds.
    \end{enumerate}
    Then, for any given vector $\vec{b}$, there is an algorithm which computes a vector $\vec{x}$ in $\overline{n}^{o(1)} Q(\rho)$ rounds such that 
    \begin{equation*}
        \| \vec{x} - \mathcal{L}(G)^{\dagger} \vec{b} \|_{\mathcal{L}(G)} \leq \eps \log n \| \vec{b} \|_{\mathcal{L}(G)^{\dagger}}.
    \end{equation*}
\end{lemma}

\subsection{Putting Everything Together}
\label{sub:everything}

In this subsection we combine the building blocks we previously developed to establish \Cref{theorem:laplacian-abstract}. The distributed Laplacian solver of \citet{DBLP:journals/corr/abs-2012-15675} is given in Algorithm \ref{algorithm:Laplacian}. We also include below the formal version of \Cref{theorem:laplacian-abstract}.

\begin{algorithm}[!ht]
\SetAlgoLined
\textbf{Input}: An undirected weighted graph $G$ \;
$G' := \specspars(G)$ \;
$(G_1, \mat{Z}_{1,1}, \mat{Z}_{1,2}, \mathcal{T}_1, G_2) := \ultraspars(G', k)$ \tcp*{\Cref{lemma:ultra_sparsification}}
$\{(G_i, \mat{Z}_{i,1}, \mat{Z}_{i,2}, \mathcal{T}_i)\}_{i=2}^t := \build(G_2, d, \eps, k)$ \;
Solve $\mathcal{L}(G) \vec{x} = \vec{b}$ via Chebyshev preconditioning \tcp*{\Cref{lemma:schur_chain}} 
\textbf{Procedure} $\build(G, d, \eps, k)$ \;
\If{$|V(G)| \leq k$}{
\textbf{return} $\emptyset$ \;
}
$(\mat{Z}_1, \mat{Z}_2, C) := \eliminate(G, d, \eps)$ \tcp*{\Cref{lemma:eliminate}}
$H := \textsc{ApproxSC}(G, C, \eps)$ \tcp*{\Cref{lemma:approxSC}}
\textbf{return} $(G, \mat{Z}_1, \mat{Z}_2, C) \cup \build(H, d, \eps, k)$\;
\caption{Distributed Laplacian Solver \cite{DBLP:journals/corr/abs-2012-15675}: $\solver(G, \eps)$}
\label{algorithm:Laplacian}
\end{algorithm}

\begin{restatable}[Full-Version of \Cref{theorem:laplacian-abstract}]{theorem}{absfull}
    \label{theorem:laplacian-abstract-full}
Consider a weighted $\overline{n}$-node graph $\overG$ for which \Cref{assumption:minor_aggregation} holds for some $Q(\rho) = O(\rho^c \mathcal{Q}(\overG))$, where $c$ is a universal constant and $\mathcal{Q} = \mathcal{Q}(\overG)$ is some parameter. Then, for any vector $\vec{b} \in \mathbb{R}^{\overline{n}}$ stored on its nodes and a sufficiently small error parameter $\eps > 0$, $\solver(\overline{G}, \eps)$ returns after $\overline{n}^{o(1)} \mathcal{Q} \log(1/\eps)$ rounds a vector $\vec{x}$ distributed on its nodes such that 
\begin{equation*}
    \| \vec{x} - \mathcal{L}(G)^{\dagger} \vec{b} \|_{\mathcal{L}(G)} \leq \eps \|\vec{b}\|_{\mathcal{L}(G)}.
\end{equation*}
\end{restatable}

The proof of this theorem is included in \Cref{appendix:main_result}. We note that a guarantee with respect to the $\mathcal{L}(G)^{\dagger}$-norm---as in \Cref{lemma:schur_chain}---can be translated to a guarantee in the $\mathcal{L}(G)$-norm. This incurs only a logarithmic multiplicative overhead since it is assumed that the weights are polynomially bounded and the dependence on $1/\eps$ is logarithmic~\citep[pp. 19--20]{Vishnoi}. Thus, the overhead is subsumed by the factor $\overline{n}^{o(1)}$. 

\section{Omitted Proofs}

In this section we include all of the proofs deferred from the main body and \Cref{section:solver}. We commence from \Cref{sec:prel}.

\subsection{\texorpdfstring{Proofs from \Cref{sec:prel}}{Section 2}}
\label{appendix:prel}

\shortcutparts*

\begin{proofsketch}
    Consider only one part $P_i$ in isolation over the network $G[P_i] + H_i$. First, we claim that there exists a simple deterministic algorithm that computes the AND-aggregate (where each node $v \in P_i$ has a input bit $\vec{x}(v)$) in $O(d)$ rounds, where each edge is used to send at most $O(1)$ messages. Concretely, any node whose input is $0$ will forward its input to all neighbors and deactivate itself. Any node which hears about the existence of an input-$0$ will forward this to all of its neighbors and deactivate itself. After $O(d)$ rounds, either all nodes have heard about the existence of a $0$ or they can conclude all inputs are $1$.

    We continue considering only one part $P_i$ in isolation. The next step is to elect a leader of $P_i$ by finding the node with the smallest ID in $P_i$; then, (1) iterate from the most significant bit of the ID to the least significant bit of the ID; (2) compute the AND-aggregate of the current bit of all the nodes' IDs; (3) if the AND-aggregate is $0$, all nodes whose current bit of the ID is $1$ will drop out.

    Putting these together we have a way of computing the aggregate of a part $P_i$ in isolation in $\tl{O}(d)$ rounds with each edge carrying $\tl{O}(1)$ messages: First, we elect a leader of $P_i$. Then, the leader initiates the computation of a spanning BFS tree of $\overG[P_i] + H_i$ by broadcasting from itself to all other nodes, and each node forwards the message to all neighbors; the neighbor from which it hears the message first is the parent in the tree. Finally, by performing a convergecast over the BFS tree, one can easily compute the aggregate in $O(d)$ rounds for a single part $P_i$.

    Finally, we have to run the algorithms on all the parts $\{ P_i \}_i$ simultaneously. However, this might incur congestion issues on some edges: algorithms associated with multiple parts want to send a message through the same edge in the same round. To prevent this, we randomly delay the start of each algorithm by selecting the delay uniformly at random between $0$ and $\tl{O}(c)$. This guarantees that the total number of messages (across all parts and all rounds) that want to cross a given edge is $\tl{O}(c)$. Hence, randomly delaying all algorithms makes the \emph{expected} number of messages crossing a given edge in a fixed round is $\Theta(1)$. By Chernoff bounds, this number is bounded by $\tl{O}(1)$ with high probability. Therefore, by simulating each round of the algorithm using $\tl{O}(1)$ rounds of communication (where each round of communication carries at most a single message across an edge), we can schedule the algorithms on all parts simultaneously~\citep{10.1145/2767386.2767417}. In turn, this allows us to complete all of the aggregates in $\tl{O}(d+c) = \tl{O}(Q)$ rounds.
\end{proofsketch}

\subsection{Proofs from \texorpdfstring{\Cref{section:congested-aggregation}}{Section 3}}
\label{appendix:proof-3}

\paLayeredToPaCongested*

\begin{proof}
Armed with \Cref{lemma:path-restr}, the claim essentially follows by leveraging \citet[Lemma 7.2 in the Full Version]{DBLP:conf/stoc/HaeuplerWZ21}. More precisely, we will have to slightly reformulate their result. \citet{DBLP:conf/stoc/HaeuplerWZ21} show how, for a given part $P_i$, one can solve the part-wise aggregate problem on $P_i$ by reducing it to a sequence of $\tl{O}(1)$-many ($1$-congested) part-wise aggregations between disjoint parts that are restricted to be simple paths $\calP'_i = \{ P'_{i, j} \}_{j=1}^{\tl{O}(1)}$ (where nodes know the paths' edges they participate in). This is sufficient to prove our result: suppose we run that reduction on all parts $P_i$ simultaneously. A single call to the part-wise aggregation on all of parts $P_i$ combined, asks to find a $\rho$-congested part-wise aggregation in which the parts $\biguplus_i \calP'_i$ are all simple paths. This is $\rho$-congested since at most $\rho$ parts $P_i$ use any node $v$, and within each such $P_i$, ever oracle call uses the node $v$ at most once (since they are disjoint).
\end{proof}

Let us briefly comment on the validity of our interpretation of \cite[Lemma 7.2]{DBLP:conf/stoc/HaeuplerWZ21}. Their statement has a few easily reconciled differences compared to our previous usage. Most notably, they compute shortcuts for a set of parts, assuming an oracle for doing so, which is a harder problem that simply computing part-wise aggregates. However, it can be easily verified that the shortcuts are used only to facilitate solving part-wise aggregations. Hence, the proof can easily be translated to require an oracle computing only part-wise aggregations.

\simGrho*

\begin{proof}
  Let us consider one round of communication in $\hatG_\rho$. Each node $v$ will simulate (learn all messages coming into) its copies $v_1, \ldots, v_\rho \in V(\hatG_\rho)$. Therefore, in each round node $v \in V(\overG)$ needs to learn all messages send to $v$'s copies $v_1, \ldots, v_\rho \in V(\hatG_\rho)$ from their neighbors in $\hatG_{\rho}$. Note that, by definition, $v$ already knows the messages sent between any two copies $v_i$ and $v_j$. Hence, in a single round $v$ can learn all messages sent to any fixed $v_i$. As a result, $\rho$ rounds of communication in $\overline{G}$ suffice to simulate a single round in $\hatG_{\rho}$.
\end{proof}

\Johan*

\begin{proofsketch}
  A simple edge-coloring algorithm presented in \cite{johansson1999simple} works by choosing a color uniformly at random from the set $\{1, \ldots, O(\Delta)\}$ for each edge. Each edge will, with constant probability, choose a color not used by its neighbors. Then, this color stays fixed and the edge drops out. Hence, after $O(\log n)$ iterations the edges will be properly colored. Implementation-wise, we can assume there is an additional node in the middle of each edge which represents that edge (this only makes the problem harder). Each edge randomly chooses and sends its color to its endpoints which, in turn, inform on whether there is a conflict. Then, the edges send back to its endpoints whether it dropped out. This iteration is then repeated until we reach a proper coloring.
\end{proofsketch}

\triv*

\begin{proof}
  First, consider any two nodes $u_i, v_j \in V(\hatG_{\rho})$ such that $\pi(u_i) \neq \pi(v_j)$, with $i, j \in [\rho]$. By construction of the layered graph $\overline{G}_i$, there exists a path of length at most $D(\overline{G})$ in the $i$-th layer of $\widehat{G}_{\rho}$ between $u_i$ to $v_i$. Thus, it follows that the (hop) distance between $u_i$ and $v_j$ is at most $D(\overline{G})$ given that $v_j$ and $v_i$, with $i \neq j$, are adjacent---the copies form a clique in the layered graph. This also implies that the distance between any two nodes $u_i$ and $u_j$, with $\pi(u_i) = \pi(u_j)$, is $1$, concluding the proof.
\end{proof}


\subsection{Useful Routines}

Before diving into the proofs of the Laplacian building blocks it will be useful to present several operations that can be performed efficiently under \Cref{assumption:minor_aggregation}. We stress that the proofs related to the Laplacian solver closely follow the approach in~\citep{DBLP:journals/corr/abs-2012-15675}. Our goal here is to translate them into our more general setting.

\begin{corollary}[Matrix-Vector Products]
    \label{corollary:matrix_vector_product}
Consider a matrix $\mat{A}$ with non-zeroes supported on the edges of an $n$-node graph $G$ which is $\rho$-minor distributed over a communication network $\overline{G}$ for which \Cref{assumption:minor_aggregation} holds for some $Q = Q(\rho)$, with values stored in the endpoints of the corresponding edges, and a vector $\vec{x} \in \mathbb{R}^{n}$ stored on the nodes $\ell(u^G)$ for $u^G \in V(G)$. Then, we can compute the vector $\mat{A} \vec{x} \in \mathbb{R}^n$ stored on the leader nodes $\ell(u^G)$ for all $u^G \in V(G)$ after $O(Q(\rho))$ rounds with high probability.
\end{corollary}

The proof of this corollary follows the one by \citet[Corollary 4.4]{DBLP:journals/corr/abs-2012-15675}, but nonetheless we state it here for completeness. 

\begin{proof}[Proof of \Cref{corollary:matrix_vector_product}]
The first step is to use \Cref{assumption:minor_aggregation} to disseminate the coordinates of vector $\vec{x}$ to the corresponding super-nodes after $Q(\rho)$ rounds; that is, for every $u^G \in V(G)$ the leader $\ell(u^G)$ passes to $\supervertex(u^G)$ the corresponding coordinate. Then, every node performs locally all the multiplications for its corresponding indices, and after $\rho$ rounds the node can deliver this information to the corresponding super-node. Observe that this is possible because $\mat{A}$ is supported on edges of $G$, and \Cref{definition:low_congestion-minor} imposes an edge-congestion bound. Finally, we invoke again \Cref{assumption:minor_aggregation} to sum all of the values of each super-node to the leader node, which gives the desired output requirement.
\end{proof}

Another important corollary of \Cref{assumption:minor_aggregation} is that we can simulate the spectral sparsification algorithm of Koutis (henceforth $\specspars$) on $G$~\citep{DBLP:conf/spaa/Koutis14}:

\begin{corollary}[Spectral Sparsification]
    \label{corollary:spectral_sparsification}
    Consider an $n$-node graph $G$ that $\rho$-minor distributes into an $\overline{n}$-node communication network $\overline{G}$ for which \Cref{assumption:minor_aggregation} holds for some $Q = Q(\rho)$. Then, for any $0 < \eps < 0.1$ we can implement the $\specspars$ algorithm of Koutis for $G$ after $O ( Q (\rho)  \log^7 n/\eps^2)$ rounds, which returns with high probability a graph $\widetilde{G}$ distributed as a $\rho$-minor into $\overline{G}$ such that 
    \begin{itemize}
        \item $\mathcal{L}(G) \approxeps \mathcal{L}(\widetilde{G})$ (Spectral approximation);
        \item $\widetilde{G}$ is a reweighted subgraph of $G$ with $O(n \log^6 n /\eps^2)$ edges in expectation.
    \end{itemize}
\end{corollary}

The proof of this corollary is fairly simple (see \cite[Corollary 4.4]{DBLP:journals/corr/abs-2012-15675}), but we give a sketch for completeness.

\begin{proof}[Proof of \Cref{corollary:spectral_sparsification}]
The $\specspars$ algorithm of Koutis iteratively uses the spanner scheme of \citet{DBLP:journals/rsa/BaswanaS07}. The latter algorithm gradually grows clusters. In particular, in each round clusters are sampled at random---a ``leader'' node determines whether the cluster is included in the sample, and then forwards the information to the rest of the cluster. Then, nodes compare the weights of their incident edges to decide whether they will join some cluster, and which incident edges will be added to the spanner. As a result, all the operations of the Baswana-Sun algorithm can be performed via the routine of \Cref{assumption:minor_aggregation}, and the claim follows. 
\end{proof}

\paragraph{Composition of Minors} We also state the extensions of \cite[Lemma 4.6]{DBLP:journals/corr/abs-2012-15675} and \cite[Corollary 4.7]{DBLP:journals/corr/abs-2012-15675}, which are related to the composition of $\rho$-minors. 

\begin{lemma}[Composing Minors]
    \label{lemma:composing_minors}
    Consider a graph $G_2$ which is $\rho_2$-minor distributed into a communication network $\overline{G}$ for which \Cref{assumption:minor_aggregation} holds for some $Q = Q(\rho)$, and a graph $G_1$ which is $\rho_1$-minor distributed into $G_2$. Then, we can compute with high probability and after $\widetilde{O}(Q(\rho_1 \rho_2))$ rounds a $(\rho_1 \times \rho_2)$-minor distribution of $G_1$ into $G$.
\end{lemma}

\begin{corollary}[Parallel Contraction]
    \label{corollary:contracting}
    Consider a graph $G$ which is $\rho$-minor distributed into a communication network $\overline{G}$ for which \Cref{assumption:minor_aggregation} holds for some $Q = Q(\rho)$. If $F$ represents a subset of the edges of graph $G$, we can obtain with high probability a $\rho$-minor distribution of $G / F$ into $\overline{G}$ in $\widetilde{O}(Q(\rho))$ rounds.
\end{corollary}

Recall that the notation $G / F$ implies the graph obtained from $G$ after contracting all the edges in the set $F \subseteq E(G)$.

\subsection{Ultra-Sparsification: Proof of \Cref{lemma:ultra_sparsification}}
\label{appendix:ultra_sparsification}

The first ingredient required for \Cref{lemma:ultra_sparsification} is a distributed version of the celebrated Alon-Karp-Peleg-West (AKPW) low-stretch spanning tree construction \cite{10.1137/S0097539792224474}, which is due to \citet{10.1145/2767386.2767440}. We commence by stating their definition of a \emph{distributed $N$-node cluster graph}, which incidentally was the basis for \Cref{definition:low_congestion-minor}.
 
 \begin{definition}[Distributed Cluster Graph, \cite{10.1145/2767386.2767440}]
    \label{definition:distributed_cluster}
 A distributed $N$-node cluster graph is a $5$-tuple $\mathcal{G} = (\mathcal{V}, \mathcal{E}, \mathcal{L}, \mathcal{T}, \psi)$ satisfying the following properties:
 
 \begin{enumerate}
     \item $\mathcal{V} = \{S_1, \dots, S_N\}$ forms a partition of the node set into $N$ clusters;
     \item $\mathcal{E}$ represents a multi-set of (weighted) edges;
     \item $\mathcal{L}$ is the set of leaders such that every cluster $S_i$ has \emph{exactly} one leader $\ell_i \in \mathcal{L}$. The ID of the leader node will also serve as the ID of the cluster, while it is assumed that nodes know the ID of their leader, as well as the size of their cluster;
     \item $\mathcal{T} = \{T_1, \dots, T_N\}$ is a set of cluster trees such that each cluster tree $T_i = (S_i, E_i)$ is a (rooted) spanning tree of the induced subgraph $G[S_i]$ of $G$, with root the leader of the cluster $\ell_i \in S_i$ (observe that this implies that the subgraph induced by each cluster $S_i$ is connected);
     \item $\psi: \mathcal{E} \mapsto E$ is a bijective function that maps every edge $\{S_i, S_j\} \in \mathcal{E}$ to some edge $\{u_i, u_j\} \in E$ connecting the corresponding clusters; i.e., it holds that $u_i \in S_i$ and $u_j \in S_j$. It is assumed that the two nodes $u_i$ and $u_j$ know that the edge $\{u_i, u_j\}$ is used to connect their respective clusters, as well as its weight.
 \end{enumerate}
 \end{definition}

Having introduced the concept of a distributed cluster graph, we state the following lemma, which is a direct corollary of the communication primitives we previously described.

\begin{lemma}
    \label{lemma:simulation-distributed_cluster}
    Let $G = (V, E)$ be an $n$-node graph $\rho$-minor distributed into an $\overline{n}$-node communication network $\overline{G} = (\overline{V}, \overline{E})$ for which \Cref{assumption:minor_aggregation} holds for some $Q = Q(\rho)$. If $\mathcal{G} = (\mathcal{V}, \mathcal{E}, \mathcal{L}, \mathcal{T}, \psi)$ is a distributed cluster graph for $G$, the following operations can be performed in $\widetilde{O}(Q(\rho))$ rounds:
    
    \begin{enumerate}
        \item The leader $\ell_i$ of each cluster $S_i$ broadcasts an $O(\log \overline{n})$-bit message to every node in $S_i$;
        \item Computing aggregation functions on $O(\log \overline{n})$-bit inputs simultaneously for all clusters, assuming the tree $T_i$ is known. 
    \end{enumerate}
\end{lemma}

\begin{proof}
The definition of a distributed $N$-node cluster graph (\Cref{definition:distributed_cluster}) implies that $\mathcal{G}$ is $1$-minor distributed over $G$, and in turn $\rho$-minor distributed into $\overline{G}$. Note that the induced distributed mapping can be obtained using $\widetilde{O}(Q(\rho) )$ rounds of communication by virtue of \Cref{lemma:composing_minors}. Thus, \Cref{assumption:minor_aggregation} leads to the desired claim.
\end{proof}

As a result, it follows that the $\textsc{SplitGraph}$ algorithm in \cite{10.1145/2767386.2767440} can be simulated on a graph $G$ which is $\rho$-minor distributed into $\overline{G}$ after $n^{o(1)} Q(\rho)$ communication rounds---under \Cref{assumption:minor_aggregation}. In particular, this observation directly gives a distributed construction of a low-stretch spanning tree:

\begin{lemma}[\cite{DBLP:journals/corr/abs-2012-15675,10.1145/2767386.2767440}]
    \label{lemma:overall_stretch}
Consider an $n$-node $m$-edge graph $G$ which is $\rho$-minor distributed into an $\overline{n}$-node communication network $\overline{G}$ for which \Cref{assumption:minor_aggregation} holds for some $Q = Q(\rho)$. Then, we can construct a spanning tree $T$ of $G$ after $n^{o(1)} Q(\rho)$ rounds such that the nodes know upper bounds on the corresponding stretches that sum to at most $m 2^{O(\sqrt{\log n \log \log n})}$.
\end{lemma}

Importantly, it turns out that the guarantee of \Cref{lemma:overall_stretch} suffices to sample edges by stretch, as implied by the following lemma.

\begin{lemma}[\cite{5671167}]
    \label{lemma:sampling-stretch}
Consider an $n$-node graph $G$ and a tree $T$ such that the nodes know upper bounds on the corresponding stretches that sum up to $\alpha$. Then, for any parameter $k$ there is a sampling procedure, implementable locally, that gives a graph $H$ which satisfies with high probability the following:

\begin{enumerate}
    \item $\mathcal{L}(G) \preceq \mathcal{L}(H) \preceq k \mathcal{L}(G)$;
    \item $H$ contains the edges of $T$ and $O(\alpha \log n /k)$ additional edges.
\end{enumerate}
\end{lemma}

The final step for establishing \Cref{lemma:ultra_sparsification} uses the parallel elimination procedure of \citet{DBLP:journals/mst/BlellochGKMPT14}, which requires a logarithmic number of rounds under the $\pram$ model of computation. Thus, we can show the following lemma:

\begin{lemma}
    \label{lemma:parallel_contraction}
Consider an $n$-node graph $H$ which is $\rho$-minor distributed into an $\overline{n}$-node communication network $\overline{G}$ for which \Cref{assumption:minor_aggregation} holds for some $Q = Q(\rho)$. Moreover, let $T$ be a spanning tree of $H$ and $W$ be the set of off-tree edges of $H$ with respect to $T$. Then, there is an algorithm which runs in $O(Q(\rho) \log n)$ rounds and returns a graph $\widehat{G}$, $1$-embeddable into $H$, satisfying the following:

\begin{enumerate}
    \item $\widehat{G}$ contains $O(|W|)$ nodes and edges;
    \item There are operators $\mat{Z}_1$ and $\mat{Z}_2$, which can be evaluated in $O(Q(\rho) \log n)$ rounds, such that 
    \begin{equation*}
        \mathcal{L}(H)^{\dagger} = \mat{Z}_1^T 
        \begin{bmatrix}
            \mat{Z}_2 & 0 \\
            0 & \mathcal{L}(\widehat{G})^{\dagger}
        \end{bmatrix}
        \mat{Z}_1.
    \end{equation*}
\end{enumerate}
\end{lemma}

With these pieces in place, \Cref{lemma:ultra_sparsification} follows directly from \Cref{lemma:overall_stretch}, \Cref{lemma:sampling-stretch}, and \Cref{lemma:parallel_contraction}.

\subsection{Sparsified Cholesky: Proof of \Cref{lemma:eliminate}}
\label{appendix:eliminate}

The proof of \Cref{lemma:eliminate} mainly relies on a distributed implementation of the \emph{Schur Complement Chain} (SCC) construction of \citet{10.1145/2897518.2897640}. In particular, the first step is to formalize a notion of almost-independence:

\begin{definition}
    \label{definition:DD}
   A matrix $\mat{M}$ is $\alpha$-diagonally dominant (henceforth $\alpha$-DD) if 
   \begin{equation*}
       \mat{M}_{i, i} \geq (1 + \alpha) \sum_{j \neq i} |\mat{M}_{i, j}|, \quad \forall i.
   \end{equation*}
   Moreover, an index set $F$ is $\alpha$-DD if $\mat{M}_{[F, F]}$ is $\alpha$-DD.
\end{definition}

An important observation is that computing the inverse $\mat{M}^{-1}[F, F]$ for an $\alpha$-DD set can be efficiently performed using a preconditioned gradient descent method. In this context, \citet{10.1145/2897518.2897640} give a simple sampling algorithm for finding ``large'' $\alpha$-DD sets given a Laplacian matrix. More precisely, their algorithm initially selects a random subset of nodes, and then it filters out these which do not met the condition of \Cref{definition:DD}. This leads to the following result:

\begin{lemma}
    \label{lemma:DD}
    Let $G$ be an $n$-node graph $\rho$-minor distributed into a communication network $\overline{G}$ for which \Cref{assumption:minor_aggregation} holds for some $Q = Q(\rho)$. Then, if $\mathcal{L}$ is the Laplacian matrix of $G$ and $\alpha \geq 0$ some parameter, there is an algorithm which computes an $\alpha$-DD subset $F$ of $\mathcal{L}$ of size at least $n/(8(1 + \alpha))$ in $O(Q(\rho) \log n)$ rounds with high probability.
\end{lemma}

Indeed, the algorithm of \citet{10.1145/2897518.2897640} determines an $\alpha$-DD subset of size $n/(8(1 + \alpha))$, while the round-complexity guarantee follows similarly to the proof in \citep[Lemma 6.7]{DBLP:journals/corr/abs-2012-15675}. Here we should note that the global aggregation steps required in the distributed implementation of \citep[Lemma 6.7]{DBLP:journals/corr/abs-2012-15675} can be trivially performed in $O(Q(1))$ rounds.

The next step is to construct an operator that approximates $\mathcal{L}^{-1}_{[F, F]}$, where $F$ is an $\alpha$-DD set, and can be efficiently applied to vectors. This is ensured by the following lemma:

\begin{lemma}[\cite{DBLP:journals/corr/abs-2012-15675}]
    \label{lemma:jacobi}
    Let $G$ be a graph $\rho$-minor distributed into an $\overline{n}$-node communication network $\overline{G}$ for which \Cref{assumption:minor_aggregation} holds for some $Q = Q(\rho)$. Moreover, let $\mathcal{L}$ be the Laplacian matrix associated with $G$, and $F$ be a subset of $V(G)$ such that $\mathcal{L}_{[F, F]}$ is $\alpha$-DD for some $\alpha \geq 4$. Then, for any vector $\vec{b}$ stored on the leaders of the super-nodes, there is an algorithm which returns in $O(Q(\rho) \log(1/\eps))$ rounds the vector $\mat{Z} \vec{b}$ stored on the same nodes, where $\mat{Z}$ is a linear operator such that 
    
    \begin{equation*}
        \mathcal{L}_{[F, F]} \preceq \mat{Z}^{-1} \preceq \mathcal{L}_{[F, F]} + \eps \cdot \schur(\mathcal{L}, F),
    \end{equation*}
    for any sufficiently small $\eps > 0$.
\end{lemma}

Again, this lemma follows from the guarantee in \cite{10.1145/2897518.2897640} regarding the \emph{Jacobi procedure}, as well as by directly adapting the distributed implementation in \cite{DBLP:journals/corr/abs-2012-15675} using \Cref{corollary:matrix_vector_product}.

\paragraph{Approximating the Schur Complement.} Moreover, $\alpha$-DD sets will be useful in the approximation of the Schur complement induced by the complementary subset of nodes. First, let us recall a combinatorial view of the Schur complement as a Laplacian matrix with weights estimated by certain random walks:

\begin{lemma}[\cite{durfee2019fully}]
    \label{lemma:laplacian-schur}
    Let $G$ be an $n$-node weighted graph and a subset of nodes $\mathcal{T}$. Moreover, consider parameters $0 < \eps < 1$ and $\mu = O(\log n /\eps^2)$. If $H$ is an initially empty graph, repeat for every edge $\{u, v\} \in E(G)$ and for $\mu$ iterations the following procedure:
    
    \begin{enumerate}
        \item Simulate a random walk starting from $u$ until it first hits $\mathcal{T}$ at some node $t_1$;
        \item Simulate a random walk starting from $v$ until it first hits $\mathcal{T}$ at some node $t_2$;
        \item Combine these two walks to get a walk $t_1 = u_0, \dots, u_{\ell} = t_2$, where $\ell$ is the length of the combined walk.
        \item Add the edge $\{t_1, t_2\}$ to $H$ with weight 
        \begin{equation*}
            \frac{1}{\mu \sum_{i=0}^{\ell-1} 1 / \vec{w}(u_i, u_{i+1})}.
        \end{equation*}
    \end{enumerate}
    Then, the resulting graph $H$ satisfies $\mathcal{L}(H) \approxeps \schur(G, \mathcal{T})$ with high probability.
\end{lemma}

It should be noted that the random walks in the lemma are implied in the usual sense, wherein a step from a node is taken with probability proportional to the edge-weights of the incident edges. In the sequel, we will compute an $\alpha$-DD set $F$ via \Cref{lemma:DD}, and then the goal will be to approximate the Schur complement on the set $\mathcal{T} = V \setminus F$. Importantly, given that $F$ is $\alpha$-DD, we can guarantee that the random walks required in \Cref{lemma:laplacian-schur} will be short in expectation. Nonetheless, a challenge that arises in the distributed context---and in particular under the $\congest$ model---is that the expected congestion of an edge may by prohibitively large. This issue will be resolved by incorporating new nodes to the terminals whenever they exceed some threshold of congestion. At the same time, however, we also have to limit the node-congestion since $G$ is minor distributed into $\overline{G}$, and we can only deal with limited congestion. This will be addressed by invoking the spectral sparsification algorithm, ensuring that the average degree, and subsequently the congestion, remains limited.    

Before we proceed with the algorithm that approximates the Schur complement, we note that we can implement the random walks of \Cref{lemma:laplacian-schur} in $\widetilde{O}(Q(\rho))$ rounds under \Cref{assumption:minor_aggregation}, as implied by the approach in~\citep{DBLP:journals/corr/abs-2012-15675}.

\begin{lemma}[\cite{DBLP:journals/corr/abs-2012-15675}]
    \label{lemma:random_walk-schur}
    Let $G$ be an $n$-node graph $\rho$-minor distributed into an $\overline{n}$-node communication network $\overline{G}$ for which \Cref{assumption:minor_aggregation} holds for some $Q = Q(\rho)$. Moreover, let $F$ be an $\alpha$-DD set, $\mathcal{T} = V \setminus F$ the set of terminals, $\eps \in (0, 1)$ some error parameter, and $\gamma \geq 1$ the congestion parameter. Then, the algorithm $\randwalkschur$ runs in $O(\alpha^{-1} \gamma Q(\rho) \log^2 n/\eps^2)$ rounds, and returns a graph $H$ along with its $(\alpha^{-1} \gamma \log n \rho)$-minor distribution into $\overline{G}$ such that 
    \begin{equation*}
        \mathcal{L}(H) \approxeps \schur(G, \widehat{\mathcal{T}}),
    \end{equation*}
    with high probability, where $\widehat{\mathcal{T}} \supseteq \mathcal{T}$ has size at most $n - |F| + O(\alpha^{-1} m \eps^{-2} \log^2 n / \gamma)$.
\end{lemma}

\begin{proof}
Let us briefly describe the $\randwalkschur$ algorithm. First, we compute the expected congestion of the family of random walks $W$ predicted by \Cref{lemma:laplacian-schur} with respect to the set of terminals $\mathcal{T}$. This is done by propagating the congestion to neighbors for $O(\alpha^{-1} \log n)$ steps. Then, we create a new set $\widehat{\mathcal{T}}$ which includes $\mathcal{T}$ as a subset, as well as all the nodes which exceeded the congestion threshold of $\gamma$ based on the estimation procedure of the previous step. Note that the congestion of a node with respect to $W$ is simply the number of times this particular node participates in some random walk of $W$. By construction, it follows that the size of $\widehat{\mathcal{T}}$ is $n - |F|$ along with all the nodes that exceeded the congestion threshold of $\gamma$. However, since $F$ is an $\alpha$-DD set it follows that the length of a random walk is $O(\alpha^{-1} \log n)$ with high probability, while for every edge we simulate $\mu = O(\log n /\eps^2)$ random walks (this is related to the concentration of the corresponding random variables, as implied by \Cref{lemma:laplacian-schur}), in turn implying that the total congestion generated by these random walks is $O(\alpha^{-1} m \eps^{-2} \log^2 n)$. As a result, only $O(\alpha^{-1} m \eps^{-2} \log^2 n/\gamma)$ nodes can have congestion more than $\gamma$, verifying the assertion regarding the size of $\widehat{\mathcal{T}}$. Next, the algorithm implements the random walks of \Cref{lemma:laplacian-schur}, but with respect to the augmented set of terminals $\widehat{\mathcal{T}}$. A Chernoff bound argument assures us that all nodes in $V \setminus \widehat{\mathcal{T}}$ will have congestion $O(\gamma)$ with high probability.

In terms of the distributed implementation, estimating the congestion can be implemented in $O(\alpha^{-1} Q(\rho) \log^2 n/\eps^2)$ rounds; this follows since every walk has length $O(\alpha^{-1} \log n)$ with high probability, and we execute $\mu = O(\log n/\eps^2)$ iterations for every edge. Also note that a single step in the procedure estimating the congestion can be implemented in $O(Q(\rho))$ rounds. Next, the generation of the random walks with respect to the augmented set $\widehat{T}$ can be performed in $O(\alpha^{-1} \gamma Q(\rho) \log^2 n/\eps^2)$ rounds with high probability; this uses the aforementioned guarantee for the congestion. The final step is to minor-distribute the graph $H$ with weights as dictated by \Cref{lemma:laplacian-schur}. This is done by assigning to the terminals the leaders of all intermediate (non-terminal) nodes. The congestion guarantee ensures that the resulting mapping is an $O(\alpha^{-1} \gamma \log n \rho)$-minor distribution into $\overline{G}$.
\end{proof}

\begin{proof}[Proof of \Cref{lemma:eliminate}]
The $\eliminate$ algorithm proceeds in $d$ rounds, initializing $\mat{M}^{(0)}$ to be an $\eps$-spectral sparsifier of $\mathcal{L}(G)$ (recall \Cref{corollary:spectral_sparsification}). In every round $i \geq 1$, (i) we compute an $\alpha$-DD set $F_i$ with $\alpha := 4$; (ii) we employ \Cref{lemma:jacobi} to have access to an operator that approximates $\mat{M}_{[F, F]}^{(i-1)}$; and (iii) we compute an $\eps$-spectral sparsifier $\mat{M}^{(i)}$ of the Schur complement $\schur(\mat{M}^{(i-1)}, \widehat{\mathcal{T}}_i)$ approximated via \Cref{lemma:random_walk-schur}; here, $\widehat{\mathcal{T}}_i = \widehat{\mathcal{T}}_{i-1} - F_i + U_i$, where $U_i$ represents the set of extra nodes added to ensure low congestion. In particular, \Cref{lemma:random_walk-schur} is invoked with congestion parameter $\gamma := 1000C \alpha^{-1} \log^8 n/\eps^4$, where $C$ is a sufficiently large constant. The sparsification algorithm of Koutis (\Cref{corollary:spectral_sparsification}) tells us that the number of edges will be $m = (n \log^6 n /\eps^2)$, in turn implying that the number of nodes drops by at least a multiplicative factor of $49/50$. 

In terms of the distributed implementation, notice that due to the selection of the parameters the approximation of the Schur complement (\Cref{lemma:random_walk-schur}) can be performed in $O(Q(\rho) \log^{10} n /\eps^6)$ rounds. Next, the spectral sparsification step can be implemented in $O(Q(\rho') \log^7 n /\eps^2)$, where $\rho' = \alpha^{-1} \gamma \log n \rho = O(\log^9 n /\eps^4) \rho$. Thus, by virtue of \Cref{assumption:minor_aggregation} we can infer that $Q(\rho') = O(\log^{c'} n/\eps^{c'}) Q(\rho)$, where $c'$ is some universal constant. Thus, after $d$ iterations the cost of these operations is bounded by $O(Q(\rho) (\log^c n /\eps^c)^d)$, where $c$ is some universal constant. Finally, the error guarantee follows directly from \Cref{lemma:jacobi,lemma:laplacian-schur,lemma:random_walk-schur}, after a direct argument bounding the accumulation of the error.
\end{proof}

\subsection{Minor Schur Complement: Proof of \Cref{lemma:approxSC}}
\label{appendix:approxSC}

We commence this subsection by introducing the notion of \emph{steady edges}, which are in a sense edges which are mutually "uncorrelated":

\begin{definition}[\cite{DBLP:journals/corr/abs-2012-15675}]
    A stochastic subset of edges $Z \subseteq E$ is called $(\alpha, \delta)$-steady with respect to an $m$-edge graph $H$ if 
    \begin{enumerate}
        \item $\mathbb{E}_Z \left[ \sum_{e \in Z} \vec{r}(e)^{-1} \vec{b}(e) \vec{b}(e)^T  \right] \preceq \alpha \mathcal{L}(H)$;
        \item For all $e \in Z$ we have $\sum_{e \neq f \in Z} \frac{|\vec{b}(e)^T \mathcal{L}(H)^{\dagger} \vec{b}(f)|}{\sqrt{\vec{r}(e)} \sqrt{\vec{r}(f)}} \leq \delta$;
        \item For all $e \in Z$ it holds that
        \begin{equation*}
            \vec{r}(e)^{-1} \vec{b}(e)^T \mathcal{L}(H)^{\dagger}
            \begin{bmatrix}
                \schur(H, \mathcal{T}) & 0 \\
                0 & 0
            \end{bmatrix}
            \mathcal{L}(H)^{\dagger} \vec{b}(e) \leq \frac{32|\mathcal{T}|}{m}.
        \end{equation*}
    \end{enumerate}
\end{definition}

In words, the first constraint ensures that no edge will be selected in the steady set with too high of a probability; the second corresponds to the localization constraint, circumscribing the (mutual) correlation of edges within the set; and the final constraint imposes a bound on the variance, and will be used in the martingale analysis (to apply Freedman's inequality). It should be stressed that the existence of such objects is highly non-trivial, and follows from the localization of electrical flows recently shown by~\citet{10.5555/3174304.3175408}. In the distributed setting, the following result will be established:

\begin{lemma}[\cite{DBLP:journals/corr/abs-2012-15675}]
    \label{lemma:find_steady}
    Let $G$ be an $n$-node $m$-edge graph $\rho$-minor distributed into $\overline{G}$ for which $\Cref{assumption:minor_aggregation}$ holds for some $Q = Q(\rho)$. For a constant $\delta \in (0, 1)$ and a subset of terminals $\mathcal{T} \subseteq V(G)$, there exists an algorithm which has access to a distributed Laplacian solver, and returns with high probability a set of at least $\delta m/(2000C \log^2 m)$ edges in expectation which is $(\delta/(1000C \log^2 m), \delta)$-steady, where $C$ is a sufficiently large constant. This algorithm requires $O(\log^2 n)$ calls to a distributed Laplacian solver to $1/\poly(n)$ accuracy on graphs that $2\rho$-minor distribute into $\overline{G}$, and $O(Q(\rho) \log^2 n)$ communication rounds. 
\end{lemma}

The first step towards establishing this lemma is to approximate the correlation of edges within some arbitrary set:

\begin{lemma}[\cite{DBLP:journals/corr/abs-2012-15675}]
    \label{lemma:approx_column}
    Let $G$ be an $n$-node graph with resistances $\vec{r}$, $\rho$-minor distributed into a communication network $\overline{G}$ for which \Cref{assumption:minor_aggregation} holds for some $Q = Q(\rho)$. Then, there is an algorithm, with access to a distributed Laplacian solver, which for any subset $W \subseteq E(G)$ and any edge $e \in W$ returns with high probability the quantity 
    
    \begin{equation*}
        \sum_{e \neq f \in W} \frac{|\vec{b}(e)^T \mathcal{L}(G)^{\dagger} \vec{b}(e)|}{\sqrt{\vec{r}(e)} \sqrt{\vec{r}(f)}}
    \end{equation*}
    to within a factor of $2$. This algorithm requires $O(\log^2 n)$ calls to a distributed Laplacian solver on graphs that $\rho$-minor distribute into $\overline{G}$ to accuracy $1/\poly(n)$, and an additional $O(Q(\rho) \log^2 n)$ communication rounds. 
\end{lemma}

The proof of this lemma follows directly from~\citep[Lemma 5.13]{DBLP:journals/corr/abs-2012-15675}, and leverages the $\ell_1$-sketch of \citet{DBLP:journals/jacm/Indyk06}. Similarly, a sketch can be employed to estimate the effect of each edge on the Schur complement:

\begin{lemma}[\cite{DBLP:journals/corr/abs-2012-15675}]
    \label{lemma:diff_apx}
    Let $G$ be an $n$-node with resistances $\vec{r}_e$, $\rho$-minor distributed into a communication network $\overline{G}$ for which \Cref{assumption:minor_aggregation} holds for some $Q = Q(\rho)$. Then, for a subset $\mathcal{T} \subseteq V(G)$, there exists an algorithm which returns with high probability an estimate of
    
    \begin{equation*}
        \vec{r}(e)^{-1} \vec{b}(e)^T \mathcal{L}(G)^{\dagger}
        \begin{bmatrix}
            \schur(G, \mathcal{T}) & 0 \\
            0 & 0
        \end{bmatrix}
        \mathcal{L}(G)^{\dagger} \vec{b}(e)
    \end{equation*}
    to within a factor of $2$. This algorithm requires $O(\log n)$ calls to a distributed Laplacian solver to accuracy $1/\poly(n)$ on graphs that $2\rho$-minor distribute into $\overline{G}$, and $O(Q(\rho) \log n)$ communication rounds. 
\end{lemma}

As a result, \Cref{lemma:find_steady} is established based on the algorithm $\textsc{FindSteady}$ in \cite{DBLP:journals/corr/abs-2012-15675}, with the round complexity guarantee following directly from \Cref{lemma:approx_column} and \Cref{lemma:diff_apx}.

The next ingredient is a pre-processing step which ensures that all the edges have leverage scores bounded away from $0$ and $1$. 

\begin{lemma}[\cite{DBLP:conf/focs/LiS18}]
    \label{lemma:split}
    Let $G$ be an $n$-node graph $\rho$-minor distributed into a communication network $\overline{G}$ for which \Cref{assumption:minor_aggregation} holds for some $Q = Q(\rho)$. If $1.1$-approximate leverage scores $\widetilde{\lev}_G(e)$ for the edges in $G$ are known, then there exists a process which returns after $\widetilde{O}(Q(\rho))$ rounds a graph $H$ such that 
    
    \begin{enumerate}
        \item $H$ is electrically equivalent to $G$;
        \item $H$ is $2\rho$-minor distributed into $\overline{G}$;
        \item All the leverage scores of edges in $H$ are between $[3/16, 13/16]$.
    \end{enumerate}
    Moreover, there exists a procedure which takes as input $G$ and returns in $O(Q(\rho))$ rounds a graph resulting from collapsing paths and parallel edges, and removing non-terminal leaves, along with a $\rho$-minor distribution into $\overline{G}$.
\end{lemma}

The distributed implementation of this lemma is fairly simple, and relies on \Cref{lemma:composing_minors}. We will also use the following lemma, which is based on the random projection scheme of \citet{DBLP:conf/stoc/SpielmanS08}:

\begin{lemma}[\cite{DBLP:journals/corr/abs-2012-15675}]
    \label{lemma:approx_leverages}
    Let $G$ be an $n$-node graph $\rho$-minor distributed into a communication network $\overline{G}$ for which \Cref{assumption:minor_aggregation} holds for some $Q = Q(\rho)$. Then, there is an algorithm with access to a distributed Laplacian solver which for all edges $e \in E(G)$ approximates the leverage score $\lev_G(e)$ to within a factor of $1 + \delta$ with high probability. This algorithm requires $O(\log n /\delta^2)$ calls to a distributed Laplacian solver on graphs which $\rho$-minor distribute into $\overline{G}$ to accuracy $1/\poly(n)$, as well as $O(Q(\rho) \log n/\delta^2)$ communication rounds.
\end{lemma}

The proof of this lemma follows directly from \citep[Lemma 5.4]{DBLP:journals/corr/abs-2012-15675}, and uses Achliopta's variant of the Johnson-Lindenstrauss lemma~\citep{DBLP:journals/jcss/Achlioptas03}. With these pieces at hand, we are ready to describe the algorithm for computing a minor Schur complement. At each iteration we first determine a set of steady edges via \Cref{lemma:find_steady}. Then, we estimate the leverage scores via the random projection scheme of \Cref{lemma:approx_leverages}, and each edge in the steady set is contracted (independently) with probability given by its (approximate) leverage scores; otherwise, the edge is deleted (for this we will use \Cref{corollary:contracting}). We also employ \Cref{lemma:split} in every iteration to ensure that leverage scores are bounded away from $0$ and $1$. This process is repeated as long as the number of edges exceeds a threshold, leading to the algorithm $\approxSC$ in~\citep{DBLP:journals/corr/abs-2012-15675}. The next theorem was shown in~\citep{DBLP:journals/corr/abs-2012-15675} using matrix martingale analysis:

\begin{lemma}[\cite{DBLP:journals/corr/abs-2012-15675}]\label{lemma:martingale}
    The algorithm $\approxSC$ takes as input a graph $G$ with a set of terminals $\mathcal{T}$ and an error parameter $\eps$, and returns with high probability a graph $H$ satisfying $|E(H)| = O(|\mathcal{T}| \log^2 n/\eps^2)$ and $\schur(H, \mathcal{T}) \approxeps \schur(G, \mathcal{T})$.
\end{lemma}

\begin{proof}[Proof of \Cref{lemma:approxSC}]
First, the algorithm only performs deletions and contractions, implying that it indeed returns a minor. Moreover, the correctness follows directly from \Cref{lemma:martingale}. To bound the requirements of the algorithm note that $\approxSC$ executes $O(\log m /\alpha)$ iterations, where $\alpha := \delta/(1000C \log^2 m) = O(\eps/\log^4 m)$, with high probability. In each iteration the dominant cost in terms of calls to a distributed Laplacian solver follows from the subroutine approximating leverage scores, which requires $O(\log n/\delta^2) = O(\log^5 n /\eps^2)$. Thus, we may conclude that $\approxSC$ requires $O(\log^{10} n/\eps^3)$ calls to a distributed Laplacian solver. The bound in terms of the round complexity follows similarly. 
\end{proof}

\subsection{Proof of \Cref{theorem:laplacian-abstract-full}}
\label{appendix:main_result}

\absfull*

\begin{proof}
The correctness of the algorithm follows directly from \Cref{lemma:ultra_sparsification,lemma:eliminate,lemma:approxSC,lemma:schur_chain}, so let us focus on the round complexity. By the guarantee of \Cref{lemma:ultra_sparsification} we know that the $\ultraspars$ routine returns a graph $G_2$ such that $|V(G_2)| = |V(G_1)| 2^{O(\sqrt{\log n \log \log n})}/k$; this follows since we have sparsified the graph in the first step. Thus, for $k = 2^{(\log \overline{n})^{2/3}}$ we can infer that $|V(G_2)| \leq |V(G_1)|/k^{1 - o(1)}$. Next, with regards to the Schur complement chain, \Cref{lemma:eliminate,lemma:approxSC} imply that $|V(G_{i+1})| \leq |V(G_i)| O(0.98^{d} \log^2 n /\eps^2)$. Hence, setting $d = 2^{(\log \log \overline{n})^2}$ and $\eps = 1/(\log \overline{n})^2$ gives us that $|V(G_{i+1})| \leq |V(G_i)| 2^{- \Theta((\log \log \overline{ n})^2)}$.

As a result, $\build$ returns a $(2^{\Theta((\log \log \overline{n})^2)}, \eps)$-Schur complement chain, which in turn implies that this chain has length $O(\log \overline{n}/(\log \log \overline{n})^2)$. Thus, \Cref{lemma:schur_chain} implies that we can use this chain to produce a solution in $\rho \overline{n}^{o(1)} Q(\rho)$ rounds, where $\rho$ represents the maximum congestion of a graph along the chain; it will be establish that $\rho = \overline{n}^{o(1)}$.

Let $f(n, \rho)$ represent the number of rounds required by $\solver$ on a graph with $n$ nodes which $\rho$-minor distributes into $\overline{G}$, and $g(n, \rho)$ the number of rounds required by $\build$ with input an $n$-node graph which $\rho$-minor distributes into $\overline{G}$. Then, if we ignore lower order terms, it follows that 
\begin{equation*}
    f(n, \rho) = \overline{n}^{o(1)} Q(\rho) + g(n/k^{1 - o(1)}, \rho),
\end{equation*}
where we used that $|V(G_2)| \leq |V(G_1)|/k^{1 - o(1)}$. Moreover, we have that

\begin{align*}
    g(n, \rho) &= O \left( (\log^c n /\eps^c)^{(\log \log \overline{n})^2} Q(\rho)  \right) + f(n, 2\rho) O(\log^{10} n/\eps^3) + g(n/2^{\Theta((\log \log \overline{n})^2)}, \rho) \\
    &= \overline{n}^{o(1)} Q(\rho) + \polylog (\overline{n}) f(n, 2 \rho) + g(n/2^{\Theta((\log \log \overline{n})^2)}, \rho),
\end{align*}
where we used that $|V(G_{i+1})| \leq |V(G_i)| 2^{- \Theta((\log \log \overline{ n})^2)}$, and we ignored lower order terms. As a result, the overall increase in congestion is $2^{O(\log \overline{n}/(\log \log \overline{n})^2)} = \overline{n}^{o(1)}$. That is, all the graphs constructed $(\overline{n}^{o(1)})$-minor distribute into $\overline{G}$. Finally, the theorem follows since by \Cref{assumption:minor_aggregation} the dependence of $Q(\rho)$ on $\rho$ is polynomial.
\end{proof}

\subsection{Proof of \texorpdfstring{\Cref{theorem:lower_bound}}{Theorem 4.7}}
\label{appendix:lb}

\lowerb*

\begin{proof}
First of all, as pointed out in \citep[Theorem 2]{DBLP:journals/corr/abs-2012-15675}, it suffices to establish the lower bound for a high-precision solver, i.e. for a sufficiently small $\eps = 1/\poly(\overline{n})$. Indeed, a low-accuracy solver ($\eps \leq \frac{1}{2}$) can always be ``boosted'' with only an $O(\log \overline{n})$ overhead in the overall complexity.

In this context, let $\overline{H}$ be the input to the spanning connected subgraph problem. We construct a resistor network $H'$ so that $\vec{r}(e) = 1$ if $e \in E(\overline{H})$, and $\vec{r}(e) = \overline{n}^4 $ for every edge $e \notin E(\overline{H})$. Moreover, let us select arbitrarily a node $v \in V(\overline{G})$. The key idea of the proof is to consider as input to the Laplacian solver a vector $\vec{b} \in \mathbb{R}^{\overline{n}}$ such that $\vec{b}(u) = -1$ for all $u \in V(\overline{G}) \setminus \{ v\}$, while $\vec{b}(v) = \overline{n}-1$. 

To analyze the output of that Laplacian system, we first analyze the simpler Laplacian system with input a vector $\vec{\chi}_{v, u} \in \mathbb{R}^{\overline{n}}$ for which the coordinate corresponding to node $v$ is $1$; the coordinate corresponding to node $u$ is $-1$; and any other coordinate is set to $0$. We recall the following well-known facts.

\begin{fact}
    \label{fact:order}
Let $\vec{\phi} = \mathcal{L}(H')^\dagger \vec{\chi}_{v, u}$. Then, for any node $w \in V(\overline{G})$ it holds that $\vec{\phi}(v) \geq \vec{\phi}(w) \geq \vec{\phi}(u)$.
\end{fact}

\begin{fact}
    \label{fact:effect-res}
Let $\vec{\phi} = \mathcal{L}(H')^\dagger \vec{\chi}_{v, u}$. Then, the $v-u$ effective resistance is such that $\res_{H'}(v, u) = \vec{\phi}(v) - \vec{\phi}(u)$.
\end{fact}

As argued in \citep{DBLP:journals/corr/abs-2012-15675}, the output of the Laplacian with input $\vec{\chi}_{v, u}$ and a sufficiently small error $\eps = 1/\poly(\overline{n})$ can be used to determine whether $v$ and $u$ are connected. Indeed, the following arguments have been extracted from their lower bound.

\begin{claim}
    \label{claim:eff-res-small}
  If $u$ and $v$ are connected in $\overline{H}$ it follows that $\res_{H'}(v, u) \leq \overline{n}-1$.
\end{claim}

\begin{proof}
It is well-known that the effective resistances satisfy the triangle inequality. Moreover, given that $v$ and $u$ are connected in $\overline{H}$, it follows that there exists a path of length at most $\overline{n} - 1$ in $H'$ so that every edge has resistance $1$ (by construction of the resistor network $H'$). As a result, the triangle inequality implies that $\res_{H'}(v, u) \leq \overline{n} - 1$.
\end{proof}

\begin{claim}
    \label{claim:eff-res-large}
    If $v$ and $u$ are not connected in $\overline{H}$ it follows that $\res_{H'}(v, u) \geq \overline{n}^2$.
\end{claim}

\begin{proof}
Suppose that $e_1, \dots, e_k$ are the edges leaving the connected component of $v$ in $\overline{H}$, for some $k \leq \overline{n}^2$. Then, the Nash-Williams inequality implies that 
\begin{equation*}
    \res_{H'}(v, u) \geq \frac{1}{\sum_{i=1}^k \frac{1}{\vec{r}(e_i)}} \geq \overline{n}^2, 
\end{equation*}
by construction of the resistor network.
\end{proof}

The next step of the proof is to incorporate in the analysis the error of the solver. To this end, let $\vec{\phi}'$ be an $\eps$-approximate solution to the linear system $\mathcal{L}(H') \vec{\phi} = \vec{\chi}_{v, u}$ in the sense that 
\begin{equation*}
    \| \vec{\phi}' - \mathcal{L}(H')^\dagger \vec{\chi}_{v, u} \|_{\mathcal{L}(H')} \leq \eps \| \vec{\chi}_{v, u}\|_{\mathcal{L}(H')^\dagger} = \eps \sqrt{\res_{H'}(v, u)}.
\end{equation*}
Moreover, since the Laplacian matrix has integer resistances up to range $\poly(\overline{n})$, it follows that for any $\vec{x}$, $\| \vec{x} \|_\infty \leq \poly(\overline{n}) \|\vec{x}\|_{\mathcal{L}}$. Thus, by setting $\eps = 1/\poly(\overline{n})$ to be sufficiently small, we have that 
\begin{equation*}
    \res_{H'}(v, u) - \frac{1}{\overline{n}} \leq \vec{\phi}'(v) - \vec{\phi}'(u) \leq \res_{H'}(v, u) + \frac{1}{\overline{n}}.
\end{equation*}
Now we will use these bounds to argue about the initial Laplacian system with input vector $\vec{b}$. By linearity, a solution of the Laplacian system with input $\vec{b}$ can be expressed as the sum of solutions of Laplacians with input $\vec{\chi}_{v, u}$ over all $u \in V(\overline{G}) \setminus \{ v \}$. Next, we let $\vec{\phi} = \mathcal{L}(H')^\dagger \vec{b}$, and $\vec{\phi}'$ be the output of the Laplacian solver for a sufficiently small $\eps = 1/\poly(\overline{n})$. Our analysis distinguishes between the following cases.

\paragraph{Case I} Suppose that $\overline{H}$ is connected. In turn, this implies that $v$ is connected with any node $u \in V(\overline{G})$. As a result, it follows from \Cref{fact:order}, \Cref{fact:effect-res} and \Cref{claim:eff-res-small} that for any node $u$,
\begin{equation}
    \label{eq:small-phi}
    \vec{\phi}'(v) - \vec{\phi}'(u) \leq (\overline{n} - 1)^2 + 1.
\end{equation}
\paragraph{Case II} In the contrary case, there must be node $u$ such that $v$ and $u$ are disconnected on $\overline{H}$. By \Cref{claim:eff-res-large} and\Cref{fact:order} 
this yields that 
\begin{equation}
    \label{eq:large-phi}
    \vec{\phi}'(v) - \vec{\phi}'(u) \geq \overline{n}^2 - 1.
\end{equation}
Thus, \eqref{eq:small-phi} and \eqref{eq:large-phi} imply that the output $\vec{\phi}'$ of the Laplacian solver contains enough information to determine whether $\overline{H}$ is connected or not since $\overline{n}^2-1>(\overline{n}-1)^2 + 1$ for any $\overline{n}\geq 2$. 

To leverage this in the $\congest$ model we proceed as follows. First, node $v$ sends to every other node in the graph its own part of the output from the Laplacian solver. This step can be clearly completed after $D(\overline{G})$ rounds. Then, each node $u$ inspects whether the value $\vec{\phi}'(v) - \vec{\phi}'(u)$ is larger than $\overline{n}^2 - 1$. In that case, node $u$ can transmit this information to the entire network; this step is easily seen to be implementable in $D(\overline{G})$. As a result, assuming that $\SQ(\overline{G}) \geq 3 D(\overline{G})$, the proof follows immediately from \Cref{theorem:shortcutquality-lb}. But the contrary case is also immediate since on \emph{any} topology solving a Laplacian system trivially requires $\Omega(D(\overline{G}))$ rounds. This completes the proof.
\end{proof}

\section{Congested Part-Wise Aggregation in the NCC Model}
\label{appendix:ncc}

The purpose of this section is to establish \Cref{lemma:conge_PA-ncc} by appropriately leveraging the machinery developed by \citet{DBLP:conf/spaa/AugustineGGHSKL19}. To this end, let us first describe one of their key communication primitives.

\paragraph{The Aggregation Problem} In the \emph{aggregation problem}, as defined by \citet{DBLP:conf/spaa/AugustineGGHSKL19}, we are given a distributive function and a set of \emph{aggregation parts} $\{ P_1, \dots, P_k \}$, with $P_i \subseteq V(\overline{G})$ for all $i$. Every aggregation part is associated with some target node $t_i \in P_i$.\footnote{In \citep{DBLP:conf/spaa/AugustineGGHSKL19} the target node does not have to belong to the corresponding aggregation part, but this additional flexibility will not be required for our purposes.} Assuming that every node holds \emph{exactly one} input value for each aggregation part of which it is a member, the goal is to let all the target nodes learn the aggregate values with respect to the associated aggregation parts. This setting allows a node to be part of multiple groups, and in particular, we let $\ell$ be the \emph{local load}: the number of groups a given node may be included in---or an upper bound thereof. In addition, if $L = \sum_{i=1}^k |P_i|$ represents the \emph{global load} of the aggregation problem, \citet[Theorem 2.3]{DBLP:conf/spaa/AugustineGGHSKL19} established the following result.

\begin{lemma}[\cite{DBLP:conf/spaa/AugustineGGHSKL19}]
    \label{lemma:aggregation-ncc}
    There exists an aggregation algorithm which solves with high probability the aggregation problem in $O(L/\overline{n} + \ell/\log \overline{n} + \log \overline{n})$ rounds of $\ncc$.
\end{lemma}

In the context of the $\rho$-congested part-wise aggregation problem (\Cref{definition:congested_part_wise_problem}), it is clear that $\ell \leq \rho$ and $L \leq \rho \overline{n}$. Thus, we are now ready to establish \Cref{lemma:conge_PA-ncc}, the statement of which is recalled below.

\lemmaCongPaNcc*

\begin{proof}
    We first employ the communication protocol of \Cref{lemma:aggregation-ncc} so that after $O(\rho + \log \overline{n})$ rounds of $\ncc$ each target node learns with high probability the aggregate values with respect to the associated aggregation parts. Next, we can essentially reverse in time the previous communication pattern, but this time using the aggregate values as determined by the target nodes. As a result, every node will know with high probability the aggregate value for each of its aggregation parts after $O(\rho + \log \overline{n})$ rounds of $\ncc$.
\end{proof}

\end{document}